\documentclass[conference,a4paper]{IEEEtran}
\normalsize
\usepackage{graphicx}
\usepackage{subcaption}
\usepackage{hyperref}
\usepackage{url}
\usepackage{array}
\usepackage[utf8]{inputenc} 
\usepackage[T1]{fontenc}
\usepackage{amsmath}
\usepackage{mathtools}
\usepackage{amsfonts}
\usepackage{amssymb}
\usepackage{amsthm}
\usepackage{enumerate}
\usepackage{cite}
\usepackage{calc}
\usepackage{color}
\usepackage{epsfig}
\usepackage{setspace}
\usepackage{pstricks}
\usepackage{cancel}
\usepackage{multirow}
\usepackage{mathrsfs}
\usepackage{algorithm}
\usepackage{algpseudocode}
\usepackage{multirow}
\usepackage{diagbox}
\usepackage{booktabs,makecell}
\usepackage{mathrsfs}
\usepackage{enumitem}
\newtheorem{defn}{Definition}
\newtheorem{thm}{{\cal T}heorem}

\newtheorem{lem}{Lemma}

\newtheorem{example}{Example}
\def\Ddots{\mathinner{\mkern1mu\raise\p@
		\vbox{\kern7\p@\hbox{.}}\mkern2mu
		\raise4\p@\hbox{.}\mkern2mu\raise7\p@\hbox{.}\mkern1mu}}
\makeatother
\newcommand{\ndiv}{\hspace{-4pt}\not|\hspace{2pt}}

\begin{document}
	\title{Multi-access Coded Caching Scheme with Linear Sub-packetization using PDAs}
\author{
	\IEEEauthorblockN{Shanuja Sasi and B. Sundar Rajan}
	\IEEEauthorblockA{  Indian Institute of Science, Bengaluru\\
		E-mail: shanuja@iisc.ac.in, bsrajan@iisc.ac.in}
}
\maketitle 
\begin{abstract}
	We consider multi-access coded caching problem introduced by Hachem et.al., where each user has access to $L$ neighboring caches in a cyclic wrap-around fashion. We focus on the deterministic schemes for a specific class of multi-access coded caching problem based on the concept of PDA. We construct new PDAs which specify the delivery scheme for the specific class of multi-access coded caching problem discussed in this paper. For the proposed scheme, the coding gain is larger than that of the state-of-the-art while the sub-packetization level varies only linearly with the number of users. Hence, we achieve a lower transmission rate with the least sub-packetization level compared to the existing schemes. 
\end{abstract}

\section{introduction}
\label{sec1}
The drastic increase in the demand for video streaming services is one of the key factors for introducing the coded caching scheme by Maddah-Ali and Niesen \cite{maddahali2015flcc}. The main objective behind the coded caching scheme is to relieve the traffic burden during peak hours by utilizing the ample cache memories available at the user ends.  The set-up consists of a central server having access to a set of $N$ files of equal size, which is connected through an error-free link to a set of $K$ users  where each user has  a cache size of $M$ files. The proposed ($K,M,N$) coded caching scheme operates in two phases. The first is the placement phase where each file is divided into $F$ equal packets and the cache memories are filled with some of these packets during the off-peak hours. The second phase is the delivery phase when the demands are revealed by the users. Once the demands are known, the server transmits coded symbols  of length $\mathcal{R}$ files over an error-free link to all the users such that all the users are able to meet their demands  from the local cache content and the transmitted symbols. The quantity $\mathcal{R}$ is referred to as the transmission rate. 

Coded caching has been extensively studied over the past few years \cite{tuninetti2016uncoded,avestimar2017tradeoff,jin2018placeOpti,vilardebo2018coded,nieson2017nonuniform,gunduz2018nonuniform,nieson2015decent,caire2016d2d}. The
value of $F$, which is known as the sub-packetization level, is directly proportional to the complexity of a
coded caching scheme.
It is well known that there is a tradeoff between the sub-packetization level $F$ and the transmission rate $\mathcal{R}$.
In the scheme introduced by Maddah-Ali and Niesen \cite{maddahali2015flcc}, which we refer as MN scheme, the sub-packetization level grows exponentially with respect to the number of users $K$, which makes it infeasible for practical implementation. Reducing the sub-packetization level of the coded caching schemes has been a major problem studied  during the past few years.

Yan et al. \cite{Yan2017PDA}, represented a coded caching scheme by an array called Placement Delivery Array (PDA) with an aim to reduce the sub-packetization level. It is shown that the MN scheme can be represented by a PDA and is optimal among regular PDAs, which is a specific class of PDAs. Although the scheme proposed by Yan et. al. has a lower sub-packetization level compared to the MN scheme at the expense of a slight increase in the transmission rate, the sub-packetization level still increases sub-exponentially with $K$.  The concept of PDA has been identified as an effective tool to reduce the sub-packetization level and since then various coded caching schemes based on the concept of PDA have been reported  \cite{cheng2020PDA,cheng2019groupingscheme,cheng2019flexiblememory,Michel2020EdgeColoring,SZG2018hypergraph,Yan2018bipartitegraph,zhong2020concatenating}.

\subsection{Multi-access Coded Caching}
 Unlike the caching models where it is assumed that each user has their own dedicated cache, in this paper, we consider a new model referred to as multi-access coded caching model which was introduced in \cite{nihkil2017multiaccess}.  In this model, as illustrated in Fig. \ref{model}, there is a central server having access to a collection of $N$ files, $\mathcal{W}=\{W_0,W_1,W_2,\ldots,W_{N-1}\}$, each of size $1$ unit, connected through an error-free link to a set of $K$ users, $\mathcal{U}=\{U_0,U_1, \ldots,U_{K-1}\}$. There are $K$ caches, $\mathcal{C}=\{C_0,C_1, \ldots,C_{K-1}\}$, each having a storage capacity $M = N \gamma$ files, where $\gamma$ is defined as the normalized cache size. Each user can access $L$ caches in a cyclic wrap-around fashion. The content stored in each cache $C_{\alpha}, \alpha \in \{0,1,\ldots K-1\}$, is denoted by $M_{\alpha}$.
The cache content accessible to the user $U_{\alpha}, \alpha \in \{0,1,\ldots K-1\},$ is denoted by $\mathcal{Z}_{\alpha}$.   The index of the file requested by the user $U_{\alpha}, \alpha \in \{0,1,\ldots K-1\},$ is denoted by $d_{\alpha}$. The demand vector is denoted by  ${\bf d} = (d_0 , d_1 , \ldots, d_{K-1} )$. 
 
 Like in the centralized coded caching scheme, multi-access coded caching scheme also operates in two phases: a placement phase and a delivery phase. In the placement phase the caches are filled with parts of the files from the servers' database. 
  In the delivery phase, each user $U_{\alpha}$ reveals their demands, which is assumed to be a file from the database. Based on the demand vector and the cache content accessible to each user, the server transmits coded symbols so that each user $U_{\alpha}$ can decode the desired file $W_{d_{\alpha}}$ using the transmissions as well as the cache content.  The overall objective of the multi-access coded caching problem is to obtain placement and delivery schemes so as to minimize the transmission rate, which is defined as the amount of data transmitted by the server in the units of files. The number of users for which each transmission is beneficial is termed as the \textit{coding gain}.
 
 {\it Notations:}  $[n]$ represents the set $\{1,2, \ldots , n\}$, $[a,b]$ represents the set $\{ a, a+1, \ldots, b \}$ while $[a,b)$ represents the set $\{a,a+1, \ldots , b-1\}$. The bit wise exclusive OR (XOR) operation is denoted by $\oplus.$  $\lfloor x \rfloor$ denotes the largest integer smaller than or equal to $x$ and $\lceil x \rceil$ denotes the smallest integer greater than or equal to $x$. 
 $a|b$ implies $a$ divides $b$ and 
 $a \ndiv b$ implies $a$ does not divide $b$, for some integers $a$ and $b$.	The transpose of any matrix \textbf{A} is represented by $\textbf{A}^T.$ For any $m \times n$ matrix  $\textbf{A} = (a_{i,j}), i \in [0,m-1],n \in [0,n-1],$  the matrix $\textbf{A}+b$, is defined as $\textbf{A}+b =(a_{i,j}+b)$.
 \subsection{Previous Results}
The multi-access coded caching scheme was introduced in  \cite{nihkil2017multiaccess} where the authors have provided a coloring based scheme for the proposed problem. 
In \cite{nikhil2020ratememory}, a new scheme, which we refer as RK  scheme, was proposed by mapping of the coded caching problem to the \textit{index coding} problem achieving a transmission rate $\mathcal{R}_{RK} $ which is less than that achieved in \cite{nihkil2017multiaccess}. 
	\begin{equation}
	\mathcal{R}_{RK}(\gamma) = 
\left\{
\begin{array}{ll}
\frac{\left (K-K\gamma L\right  )^2}{K}, & \forall \gamma  \in \left \{\frac{k}{K} : k \in [0,\left \lfloor \frac{K}{L} \right \rfloor] \right \} \\
0, &  \textit{for     } \gamma = \left \lceil \frac{K}{L} \right \rceil \frac{1}{K}
\end{array}
\right.
	\end{equation}
The sub-packetization level required for RK  scheme is $F_{RK}=\frac{1}{\gamma } {K-K\gamma (L-1)-1 \choose K\gamma -1}$.

 A lower bound $ \mathcal{R}_{lb}$ on the optimal transmission rate-memory trade-off was also derived in \cite{nikhil2020ratememory} for any multi-access coded caching problem when $L \geq \frac{K}{2}$.
For some special cases, namely for $L \geq \frac{K}{2}$, and when $L=K-1, L=K-2, L=K-3$ when $K $ is even, and
$L=K-\frac{K}{g}+1$ for some positive integer $g$, an achievable scheme was proposed separately in \cite{nikhil2020ratememory}, which is optimal.
\begin{align*}
 \mathcal{R}_{lb}(\gamma ) = \begin{cases*}
 K- \left [  K- \frac{(K-L)(K-L+1)}{2K} \right ] K \gamma , & \text{if $0\leq \gamma  \leq \frac{1}{K}$} \\
 	\frac{(K-L)(K-L+1)}{2K}\left (2-K \gamma  \right ), & \text{if $\frac{1}{K}\leq \gamma  \leq \frac{2}{K}$} \\
 	0, & \text{if $\gamma  \geq \frac{2}{K}$}
 \end{cases*}
\end{align*}

In \cite{parinello2019multiacessgains}, the authors have studied two special cases, the first case is when $K\gamma =2$, and the second case is when $L=\frac{K-1}{K\gamma }$ for any $K\gamma$. For $K\gamma =2$, the authors have proposed a novel coded caching scheme, which we call as SPE scheme, that achieves a coding gain that exceeds $2$ with a sub-packetization level of $F_{SPE} =\frac{K(K-2L+2)}{4}$. The transmission rate achieved using the SPE scheme is \begin{equation} \label{rate SPE}
\mathcal{R}_{SPE}=\frac{X_1+X_2}{S},
\end{equation}
where $X_1$, $X_2$ and $S$ are defined in Eq. (2) in \cite{parinello2019multiacessgains}.
The coding gain always exceeds $3$ and approaches $4$ for some values of $K$ and $L$. For the second case considered in \cite{parinello2019multiacessgains}, i.e., when $L=\frac{K-1}{K\gamma }$ for any $K\gamma $, the authors have provided an achievable scheme which is optimal.

In \cite{SR2020improvedrate}, the authors have proposed an improved scheme for any $\gamma \in \left \{\frac{k}{K}: \gcd(k,K)=1,  k\in \left [1,K \right ]\right \}$, achieving a rate which is less than or equal to $\mathcal{R}_{RK}(\gamma).$ 

In \cite{Caire2020noveltransformation}, the authors have used a \textit{novel transformation} approach to extend the MN scheme to multi-access caching schemes, such that the resulting scheme has the maximum local caching gain and the same coding gain as the related MN scheme. We refer to the scheme proposed in \cite{Caire2020noveltransformation} as NT scheme. The coding gain achieved using the NT scheme is $K\gamma +1$ with a sub-packetization level of $F_{NT}=K {K-K\gamma(L-1) \choose K\gamma}$. 
The transmission rate achieved using the NT scheme is given by
 \begin{equation}
\mathcal{R}_{NT} = \frac{K-K\gamma L}{K\gamma+1}.
\end{equation}

In \cite{nikhil2020structured}, the multi-access coded caching problem is divided into a number of special class of index coding problems termed as \textit{Structured Index Coding} problem and using the solutions obtained for the structured index coding problem, a new scheme, which we refer as  NK scheme, was proposed for the multi-access coded caching problem which achieves the following rate:
\begin{equation} \label{rate NK }
\mathcal{R}_{NK } = \frac{\sum_{{\bf b} \in \mathcal{B} } \min\{2(K-K\gamma L) +K\gamma-1-\tilde{\bf b},K\}}{F_{NK }(K\gamma+1)},
\end{equation}
where $\mathcal{B}$ is the collection of all the weak $K\gamma +1$ compositions of $K-K\gamma L-1$, $\tilde{\bf b}$ denotes the maximum component in the vector ${\bf b}  \in \mathcal{B}$ and $F_{NK }=\frac{1}{\gamma } {K-K\gamma (L-1)-1 \choose K\gamma -1}$ denotes the sub-packetization level.
In \cite{MR2020linearsubpacketization}, the authors have proposed a scheme for $\gamma=\frac{1}{K}$ when $L\geq \frac{K}{2}$, which achieves linear sub-packetization level with a slightly higher load than the above schemes.
  Multi-access coded caching scheme with the number of users not equal to the number of caches was studied in \cite{KMR2020CRD} where the authors have identified a special class of resolvable designs called cross resolvable designs.

	In \cite{SR2020improvedrate}, we have studied the cases when $\gamma \in \left \{\frac{k}{K}: \gcd(k,K)=1,  k\in \left [1,K \right ]\right \}$. In this work we have taken up a particular case, when $\gamma \in \left \{\frac{k}{K}: k | K, (K-kL+k)|K, k\in \left [1,K \right ]\right \}$, which is not covered in \cite{SR2020improvedrate}. Also, for the cases studied in \cite{SR2020improvedrate}, the worst case sub-packetization level required is $K^2$ while in this work, the sub-packetization level required is $K$.
%

\begin{figure}[!t]
	\centering
	\includegraphics[width=17pc]{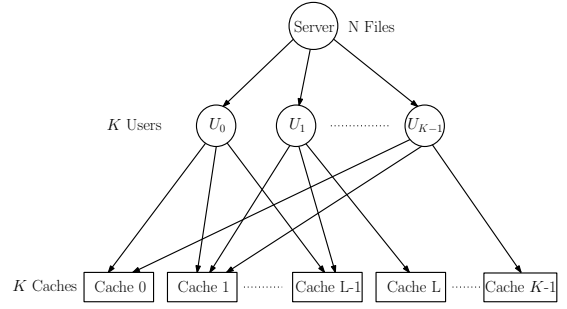}
	\caption{Multi-access Coded Caching Network \cite{nihkil2017multiaccess} consisting of a central server, $K$ users, and $K$ caches where each user is connected to $L$ neighboring caches.}
	\label{model}
\end{figure}

\subsection{Our Contributions}
Our contributions in this paper are summarized as follows.
\begin{itemize}
	\item We construct a new class of PDA which we call as \textit{$t$-cyclic $g$-regular PDA}, which is used for providing the delivery algorithm in our scheme.
	\item We prove that the advantage of the proposed scheme is two-fold, in terms of the coding gain as well as the sub-packetization level compared to NT, RK  and SPE schemes. The sub-packetization level $K$ required for our scheme is $F_{new}=K$, which varies linearly with the number of users while the sub-packetization level varies sub-exponentially with respect to the number of users in NT, RK  and SPE schemes.
\end{itemize}

\section{Main Result}
\label{sec2}	
We discuss our main result in this section. We characterize our result in Theorem \ref{thm1}.
\begin{thm} 
	\label{thm1}
	Consider a  multi-access coded caching scenario with $N$ files, and $K$ users, each having access to $L$ neighboring caches in a cyclic wrap-around way, with each cache having a normalized capacity of $ \gamma$, where $\gamma \in \left \{\frac{k}{K}: k | K, (K-kL+k)|K, k\in \left [1,K \right ]\right \}$. The following transmission rate $\mathcal{R}_{new}(\gamma)$ is achievable.

	\begin{equation}
				\mathcal{R}_{new}(\gamma) = \frac{(K-kL)(K-kL+k)}{2K}
	\end{equation}
\end{thm}
 Comparison of our scheme with the state-of-the-art is done in Table \ref{tab1}.
  The transmission rate vs $L$ plot vs $\gamma$  is  obtained for $K=24$ in Fig. \ref{fig: pda} and the sub-packetization level vs $L$ vs $\gamma$ plot is  obtained for the same value of $K$ in Fig. \ref{fig: subpacket}.
   It can be observed that the transmission rate achieved using our scheme for the points considered in Theorem \ref{thm1}  is less compared to that achieved using the RK, and SPE schemes. It can also be observed that the transmission rate achieved using our scheme for all the points considered in Theorem \ref{thm1}, except when $L=3, k=6$ and $L=4,k=4$, is less compared to that achieved using the NT scheme.The sub-packetization level required is less for our scheme compared to  the NT, RK, and SPE schemes.  
The placement scheme and the delivery algorithm achieving the rate claimed in Theorem \ref{thm1} is given in Section \ref{proof}.
   \begin{figure*}[!t]
 	\centering
 	\includegraphics[scale=0.35]{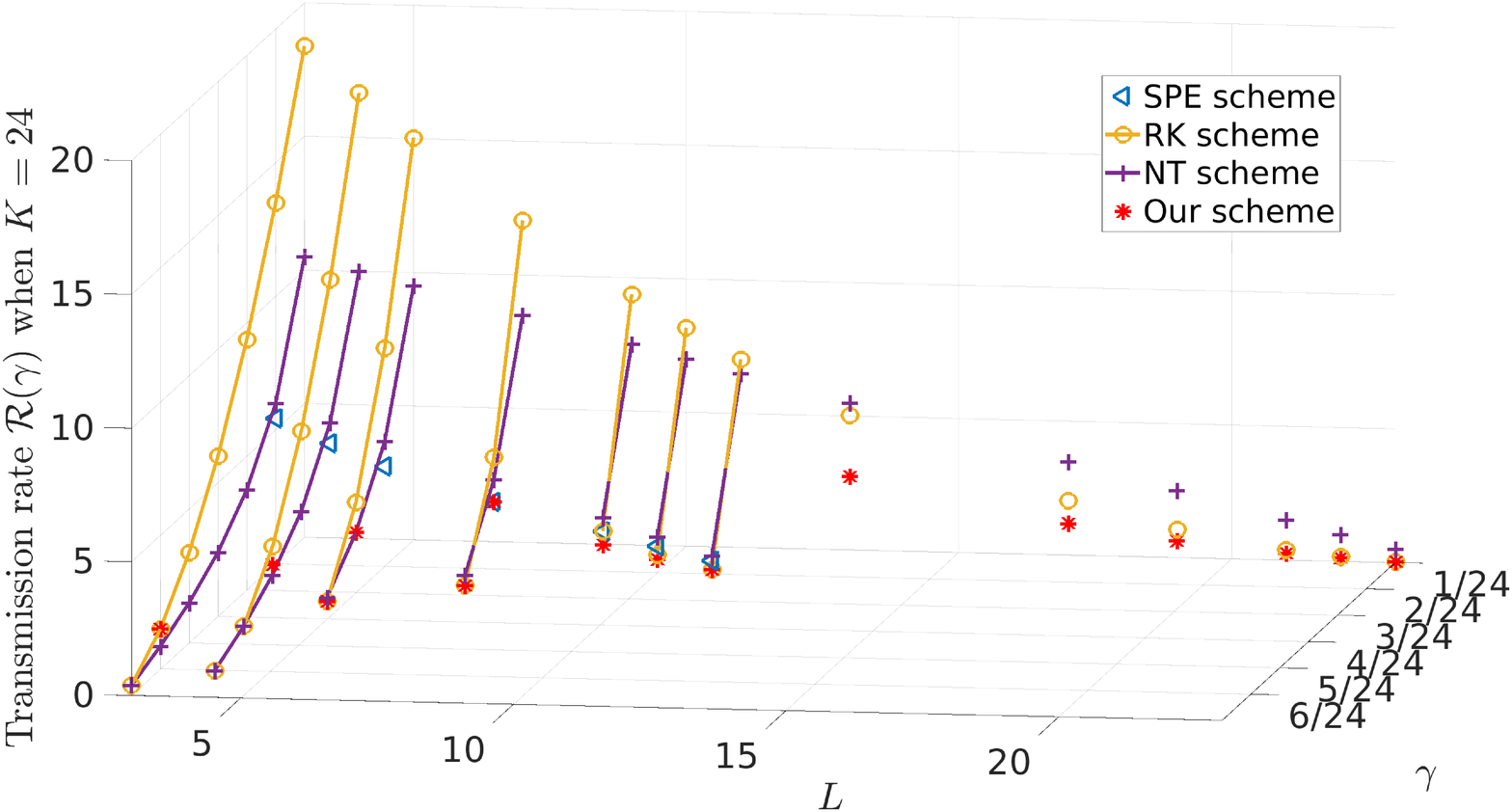}
 	\caption{Transmission rate vs $L$ vs $\gamma$ plot when $K=24$.}
 	\label{fig: pda}
 \end{figure*}
 \begin{figure*}[!t]
	\centering
	\includegraphics[scale=0.35]{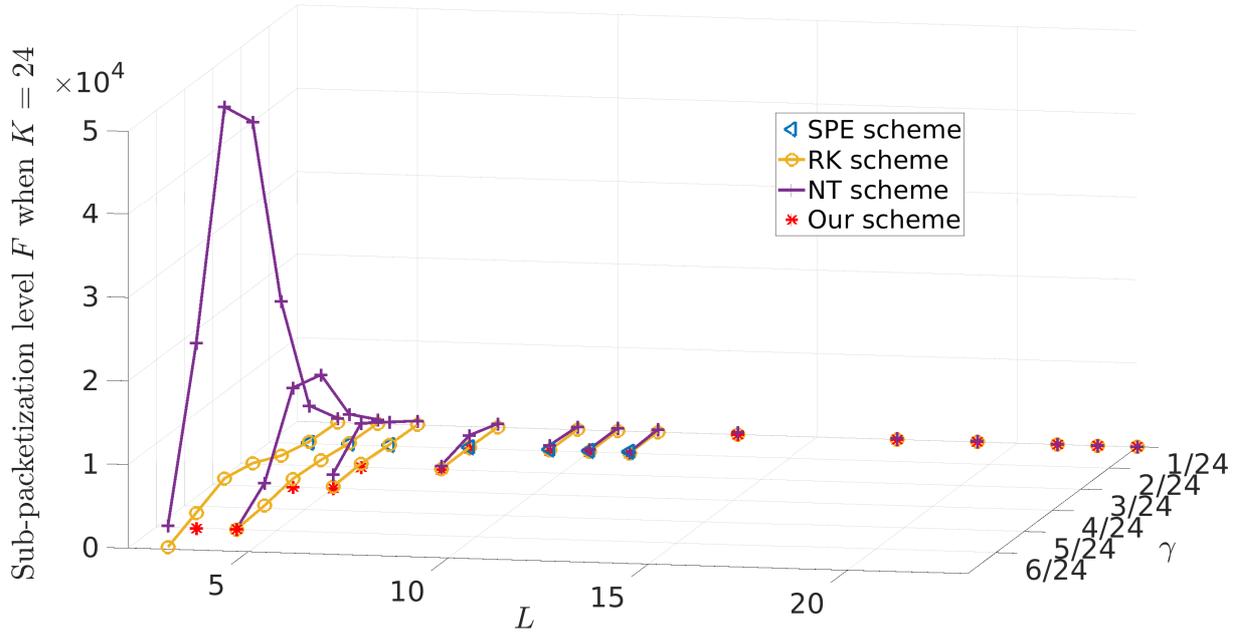}
	\caption{Sub-packetization level vs $L$ vs $\gamma$ plot  when $K=24$.}
	\label{fig: subpacket}
\end{figure*}
\begin{table*}
	\begin{center}
		\begin{tabular}{ | c|c|c  |c|}
			\hline	
		& \textit{Sub-packetization Level}  &  \textit{Coding Gain} & \textit{Transmission Rate}   \\ \hline \hline
			NT scheme \cite{Caire2020noveltransformation}& $K{K-kL+k \choose k}$& $k+1$& $\frac{K-kL}{k+1}$ \\ \hline
	RK  scheme	\cite{nikhil2020ratememory}&${K-kL+k-1 \choose k-1} \frac{K}{k}$ & $\frac{K}{K-kL}$&$\frac{(K-kL)^2}{K}$ \\ \hline
		SPE scheme (for $k=2$)	\cite{parinello2019multiacessgains}&$(\frac{K(K-2L+2)}{4})$ & $3$ to $4$& $\mathcal{R}_{SPE}$ as in (\ref{rate SPE}) \\ \hline
		NK scheme \cite{nikhil2020structured} 	&${K-kL+k-1 \choose k-1} \frac{K}{k}$ & $\frac{\mathcal{R}_{NK }}{K-kL}$&$\mathcal{R}_{NK}$ as in (\ref{rate NK }) \\ \hline
		Our scheme & $K$& $\frac{2K}{K-kL+k}$& $\frac{(K-kL)(K-kL+k)}{2K}$ \\ \hline
		\end{tabular}
	\end{center}
	\caption{Comparison of our scheme with the state-of-the-art when  $\gamma \in \left \{\frac{k}{K}: k | K, (K-kL+k)|K, k\in \left [1,K \right ]\right \}$.}
	\label{tab1}
\end{table*}
\subsection{Comparison with the NT Scheme}
The coding gain achieved in our scheme is $\frac{2K}{K-kL+k}$ while that is $k +1$ in the NT scheme. We prove in Lemma \ref{lem NT coding gain} that the coding gain achieved in our scheme is more than that achieved using the NT scheme, if $L >\frac{K(k-1)}{k(k+1)}+1$.  The sub-packetization level required for our scheme is $F_{new}=K$ while it is $F_{NT} =K {K-kL+k \choose k}$ for the NT scheme. Hence the sub-packetization level grows linearly with the number of users in our scheme while it grows sub-exponentially with respect to the number of users in the NT scheme. So, the transmission rate achieved in our scheme is less than that achieved using the NT scheme, if $L >\frac{K(k-1)}{k(k+1)}+1$. Moreover, the coding gain increases as $L$ increases in our scheme while the coding gain is independent of $L$ in the NT scheme. Hence, as $L$ increases the gap between the transmission rate between our scheme and the NT scheme increases. Our scheme is better in terms of both coding gain as well as sub-packetization level compared to the NT scheme. 
\begin{lem} \label{lem NT coding gain}
	For the cases considered in Theorem \ref{thm1}, the coding gain achieved in our scheme, $\frac{2K}{K-kL+k}$, is more than the coding gain, $k+1$, achieved using the NT scheme, if $L >\frac{K(k-1)}{k(k+1)}+1$.
\end{lem}
\begin{proof}
	Assume that the coding gain achieved in our scheme is less than or equal to that achieved using the NT scheme if $L >\frac{K(k-1)}{k(k+1)}+1$, i.e.,
	\begin{align*}
	k+1& \geq \frac{2K}{K-kL+k}\\
	\Rightarrow (k+1)(K-kL+k) &\geq  2K \\
\Rightarrow Kk+K-k(k+1)(L-1) &\geq  2K \\
\Rightarrow k(k+1)(L-1) &\leq  K(k-1) \\
\Rightarrow L &\leq  \frac{K(k-1)}{k(k+1)}+1
	\end{align*}
	This contradicts our assumption that $L >\frac{K(k-1)}{k(k+1)}+1$. Therefore the coding gain achieved in our scheme is more than that achieved using the NT scheme, if $L >\frac{K(k-1)}{k(k+1)}+1$.
\end{proof}
\subsection{Comparison with the RK  Scheme}
For the considered cases, the coding gain achieved in our scheme is $\frac{2K}{K-kL+k}$ while that is $\frac{K}{K-kL}$ in the RK  scheme. We prove in Lemma \ref{lem RK  coding gain} that our scheme is better in terms of coding gain as compared to the RK  scheme. The sub-packetization level required for our scheme is $F_{new}=K$ while it is  $F_{RK}=\frac{1}{\gamma } {K-K\gamma (L-1)-1 \choose K\gamma -1}$ for the RK  scheme. Even though the sub-packetization level required for the RK  scheme is less compared to the NT scheme, it still grows sub-exponentially with respect to the number of users while the sub-packetization level
grows only linearly with the number of users in our scheme. Hence, compared to the RK  scheme, our scheme is better in terms of both coding gain and sub-packetization level.
\begin{lem} \label{lem RK  coding gain}
	For the cases considered in Theorem \ref{thm1}, the coding gain achieved in our scheme, $\frac{2K}{K-kL+k}$, is greater than or equal to the coding gain, $\frac{K}{K-kL}$, achieved using the RK  scheme.
\end{lem}
\begin{proof}
	The coding gain achieved in our scheme is less than that achieved using the RK  scheme, only if
	\begin{align*}
	\frac{K}{K-kL}& >  \frac{2K}{K-kL+k}\\
	\Rightarrow K-kL+k &>   2K-2kL \\
\Rightarrow
	kL+k &> K \\
	\Rightarrow L &> \frac{K}{k}-1
	\end{align*}
	If $L\geq \frac{K}{k}$, then all the users can access all the sub-files of each file and the transmission rate is zero.  Therefore the coding gain achieved in our scheme is greater than or equal to that achieved using the RK  scheme, for the cases considered in Theorem \ref{thm1}.
\end{proof}
\subsection{Comparison with the SPE Scheme}
For the case when $k=2$, the SPE scheme achieves a coding gain that always exceeds $3$ and approaches $4$ for some values of $K$ and $L$. For our scheme, since $\frac{2K}{K-kL+k}$ is an integer (we have assumed that $K-kL+k $ divides $K$), the coding gain achieved when $k=2$ is always greater than or equal to $k+2=4$ (as proved in Lemma \ref{lem NT coding gain}). Also, the sub-packetization level required for SPE scheme is $F_{SPE} =\frac{K(K-2L+2)}{4}$ which is $\frac{(K-2L+2)}{4}$ times  more than the sub-packetization level required for our scheme. Hence, compared to the SPE scheme also, our scheme is better in terms of both coding gain and sub-packetization level for the cases considered in Theorem \ref{thm1}.

\subsection{Comparison with the NK Scheme}
For NK scheme, the rate expression $\mathcal{R}_{NK }$ is given by (\ref{rate NK }). Due to the complexity in the expression, we are not able to compare this with our rate analytically. From Example \ref{exmp macc}, we observe that the transmission rate when $K=12,k=2,L=4$ in the NK scheme is slightly less compared to our scheme. However, we gain in terms of the sub-packetization level required at the expense of a slight increase in the rate. 
The sub-packetization level required for our scheme is $F_{new}=K$ while it  is $F_{NK }=\frac{1}{\gamma } {K-K\gamma (L-1)-1 \choose K\gamma -1}$ for the NK scheme. 

\section{placement and delivery scheme}
\label{proof}
In this section, initially we present our placement scheme. After that we define a new class of PDA which is used for providing our delivery algorithm. Finally
we present our delivery scheme to prove  Theorem \ref{thm1}. 
\subsection{Placement Scheme}
In the placement phase, we split each file $W_n, n=[0,N)$, into $K$ disjoint sub-files, $W_n=\{W_{n,\alpha}: \alpha \in [0,K)\}$. Each cache $C_{\alpha}, \alpha \in [0,K),$ is filled as follows:
$$M_{\alpha}= \{W_{n,(k \alpha +j) \text{ mod } K}:  j \in [0,k),   n \in [0,N)\}.$$
 Each cache stores $k $ sub-files from all the files, where each sub-file is of size $\frac{1}{K}$. Hence, $M=\frac{k N}{K} =N \gamma $, thus meeting our memory constraint.

The placement is done in such a way that we first create a list of size $1 \times kK$ by repeating the sequence $\{0, 1, \ldots, K-1 \}$, $k$ times, i.e., $\{0,1, \ldots,K-1,0,1,\ldots,K-1,0,1,2,\ldots, K-1,\ldots\}.$ We fill the caches by taking $k$ items sequentially from the list.  Hence, the first cache is filled with the first $k$ items, the second cache with the next $k$ items and so on. 

Each user can access $L$ neighboring caches and each cache stores $k$ consecutive sub-files of each file. 
If $L \geq \left \lceil \frac{K}{k} \right \rceil$, then the user has access to all the sub-files of each of the files.
For the case under consideration, each user has access to $kL$ consecutive sub-files of each file since the content in any consecutive $L$ caches are disjoint from one another if $L < \left \lceil\frac{K}{k} \right \rceil$.
That is, for each user $U_{\alpha} , \alpha\in [0,K)$, the  accessible cache content is $\{W_{n,(k \alpha +i) \text{ mod } K} : i \in [0, kL),   n \in [0,N)\}$.		

Each user's demand of one file among the $N$ files from the central server is revealed after the placement phase, i.e, the demand vector ${\bf d}$ is revealed. Our delivery scheme is based on the concept of PDA. So, before presenting the delivery scheme we provide a review on PDAs and after that we define a new class of PDAs which is required for our delivery scheme. 

\subsection{Review on Placement Delivery Arrays}

\begin{defn}	\label{def: PDA}
	(\cite{Yan2017PDA}) For positive integers $K,F,Z$ and $S$, an $F \times K$ array $\textbf{P} = (p_{i,j}), i \in [0, F ), j \in [0, K )$, composed of a specific symbol ``$\star$'' called star and $S$ positive integers 
	$0,1,\ldots ,S-1,$ is called a $(K,F,Z,S)$ placement delivery array (PDA) if it satisfies the following three conditions.
	
	\begin{enumerate}[label={C\arabic*:}]
	
		\item the symbol $\star$ occurs exactly $Z$ times in each column
		\item  each integer occurs at least once in the array
		\item for any two distinct entries $p_{i_1,j_1}$ and $p_{i_2,j_2}$, such that $p_{i_1,j_1}=p_{i_2,j_2}=s$ is an integer only if 
		\begin{enumerate}
			\item $i_1 \neq i_2,j_1 \neq j_2,$ i.e., they lie in distinct rows and distinct columns and
			\item $p_{i_1,j_2}=p_{i_2,j_1}=\star$, i.e., the corresponding $2 \times 2$ sub-array formed by the rows $i_1$, $i_2$ and the columns $j_1, j_2$ must be of the following form
				 $ \left( \begin{array}{cc}
			s & \star \\
			\star & s \end{array}  \right) \text{ or } \left( \begin{array}{cc}
			\star & s \\
			s & \star \end{array}  \right) $
		\end{enumerate}
	\end{enumerate}
\end{defn}
\begin{defn}(\cite{Yan2017PDA})
	\label{def regular PDA}
	An array \textbf{P} is said to be a $g$-regular $(K, F, Z, S)$ PDA, if it satisfies the conditions \textit{C1, C3,} and the following condition
	\begin{enumerate}[label={$C2'$:}]
		\item Each integer appears $g$ times in \textbf{P} where $g$ is a constant.
	\end{enumerate}

\end{defn} 
\begin{thm} (\cite{Yan2017PDA})\label{thm PDA delivery}
	For a given $(K,F,Z,S)$ PDA, $\textbf{P} = (p_{i,j}), i \in [0, F ), j \in [0, K )$, we can obtain a ($K,M,N$) coded caching scheme having sub-packetization level $F$ with $\frac{M}{N} =\frac{Z}{F}$ using Algorithm \ref{algo1}. For any demand vector $\textbf{d},$ the demands of all the users are met at the transmission rate of $\mathcal{R}=\frac{S}{F}$.
\end{thm}
In a $(K, F, Z, S)$ PDA \textbf{P}, rows correspond to the packets of each file and columns correspond to the users. In any column $j \in [0,K)$, if $ p_{i,j} = \star,$ then user $U_j$ has access to the $i^{th}$ packet of all the files. If $ p_{i,j} = s$ is an integer, then it implies the $i^{th}$ packet of none of the files is accessible to the user $U_j$. The condition \textit{C1} implies some $Z$ packets of all the files are accessible to each user. The condition \textit{C2} implies that the number of symbols transmitted by the server is exactly $S$ since XOR of the requested packets indicated by $s$ is broadcast by the server during the slot for $s$. Hence the transmission rate is $\frac{S}{F}$. The condition \textit{C3} makes sure that each user can get the demanded packet, since all the other packets in the coded symbol are available at its cache.
The three conditions in the Defintion \ref{def: PDA} of the PDA guarantees all the users retrieve the requested files.
\begin{algorithm}
	\caption{Coded caching scheme based on PDA in \cite{Yan2017PDA}}
	\label{algo1}
	\begin{algorithmic}[1]
		\Procedure{PLACEMENT}{$\textbf{P},\mathcal{W}$}
		\State Split each file $W_n \in \mathcal{W} $ into $F$ packets, i.e., $W_n=\{W_{n,i} : i \in [0,F)\}$
		\For{\textit{$j \in [0,K)$}}
		\State \textit{$\mathcal{Z}_j \gets \{W_{n,i} : p_{i,j}=\star, \forall n \in [0,N)\}$}
		\EndFor
		\EndProcedure
		\Procedure{DELIVERY}{$\textbf{P},\mathcal{W}, \textbf{d}$}
		\For{\textit{$s = 0,1,\ldots, S-1$}}
		\State Server sends $\mathcal{Y}_s=\bigoplus_{p_{i,j}=s,i\in [0,F),j \in [0,K)}W_{d_j,i}$
		\EndFor
		\EndProcedure
	\end{algorithmic}
\end{algorithm}
\subsection{A New Class of PDAs }
In this sub-section we introduce a new class of PDA, termed as $t$-cyclic $g$-regular $(K,F,Z,S)$ PDA which is defined in Definition \ref{def new PDA}.
\begin{defn}($t$-cyclic $g$-regular $(K,F,Z,S)$ PDA) \label{def new PDA}
	In a $g$-regular $(K,F,Z,S)$ PDA \textbf{P}, if all the $Z$ stars in each column occur in consecutive rows and if the position of stars in each column in \textbf{P} is obtained by cyclically shifting the previous column downwards by $t$ units, then it is called as \textit{$t$-cyclic $g$-regular $(K,F,Z,S)$ PDA}. 
\end{defn}
\begin{example} \label{exmp: PDA 12x6}
	Consider the following $12 \times 6$ array $\textbf{P}_{12 \times 6}$. 
\begin{equation}\label{PDA 12x6}
	\textbf{P}_{12 \times 6}= 
\left( \begin{array}{ccc|ccccc}
		\star &0 & 1& \star& \star &\star \\
	\textcolor{red}{\star} &\textcolor{red}{3} & \textcolor{red}{4}& \textcolor{red}{\star}& \textcolor{red}{\star} &\textcolor{red}{\star} \\ 
		\star & \star & 2 &0 &\star &\star\\
		\textcolor{red}{\star} & \textcolor{red}{\star} & \textcolor{red}{5} &\textcolor{red}{3} &\textcolor{red}{\star} &\textcolor{red}{\star}\\ 
		\star & \star & \star& 1&2&\star \\ 
		\textcolor{red}{\star} & \textcolor{red}{\star} & \textcolor{red}{\star}& \textcolor{red}{4}&\textcolor{red}{5}&\textcolor{red}{\star} \\  \hline
		\star &\star & \star &\star &0&1\\
		\textcolor{red}{\star} &\textcolor{red}{\star} & \textcolor{red}{\star} &\textcolor{red}{\star} &\textcolor{red}{3}&\textcolor{red}{4}\\ 
		0&\star&\star &\star&\star &2\\
			\textcolor{red}{3}&\textcolor{red}{\star}&\textcolor{red}{\star} &\textcolor{red}{\star}&\textcolor{red}{\star}& \textcolor{red}{5}\\
		1&2&\star &\star &\star &\star\\
			\textcolor{red}{4}&\textcolor{red}{5}&\textcolor{red}{\star} &\textcolor{red}{\star} &\textcolor{red}{\star} &\textcolor{red}{\star}
		\end{array}\right)
		\end{equation}
	
	There are $Z=8$ stars in each column and there are $6$ integers in the array $\textbf{P}_{12 \times 6}$ where each integer $s \in [0,5]$ appears $4$ times. For each integer $s \in [0,5]$, it can be verified that the condition \textit{C3} in Definition \ref{def: PDA} is also satisfied by the array $\textbf{P}_{12 \times 6}$. Hence  the array $\textbf{P}_{12 \times 6}$ represents a $4$-regular $(6,12,8,6)$ PDA. It can be observed that all the $8$ stars in each column occur in consecutive rows and the position of stars in each column can be obtained by cyclically shifting the previous column by $2$ units down. Therefore, it is a 2-cyclic 4-regular $(6,12,8,6)$ PDA.
\end{example}

Our delivery algorithm is based on using $t$-cyclic $g$-regular PDAs. Before providing the general delivery algorithm, we illustrate the idea of using $t$-cyclic $g$-regular PDAs  for solving multi-access coded caching problems with the help of an example.

\begin{example} \label{exmp macc}
	Consider the case $N=K=12,k=2,L=4$. The server has access to $12$ files: $\mathcal{W}=\{W_n : n \in [0,12)\}$, and each file $W_n,n\in [0,12)$, is divided into $12$ sub-files: $W_n=\{W_{n,j} : j \in [0,12)\}$. Each cache $C_{\alpha}, \alpha \in [0,12),$ is filled as $M_{\alpha}= \{W_{n,(2 \alpha +j) \text{ mod } 12}:  j \in \{0,1\},   n \in [0,12)\}$.
	
	Each user $U_{\alpha}, \alpha \in [0,12)$, has access to all the caches in the set $\{C_{\alpha},C_{(\alpha+1) \text{ mod } 12},C_{(\alpha+2) \text{ mod } 12},C_{(\alpha+3) \text{ mod } 12}\}$.
	Hence, for each user $U_{\alpha},$ the accessible cache content is $\mathcal{Z}_{\alpha} =\{W_{n,(2\alpha + i) \text{ mod } 12} : i \in [0,8),n \in [0,12)\}$.

	Consider the $2$-cyclic $4$-regular $(6,12,8,6)$ PDA, $\textbf{P}_{12 \times 6}$, obtained in Example \ref{exmp: PDA 12x6}. 		 
	For the array $\textbf{P}_{12 \times 6}$ with alphabets $[0, 6) \cup \star$, define an array $\textbf{P}_{12 \times 6}+6 = (p_{i,j}+6)$, where $ \star +6= \star.$   The number of occurrences of each integer in $\textbf{P}_{12 \times 6}+6$ is exactly equal to that of $\textbf{P}_{12 \times 6}$. Also, the number of stars in each column in $\textbf{P}_{12 \times 6}+6$ is exactly same as in $\textbf{P}_{12 \times 6}$.		
	Now, we construct a new array
	
	\begin{align}
	& \textbf{P}_{12 \times 12} =\left( \textbf{P}_{12 \times 6} \hspace{0.5cm} \textbf{P}_{12 \times 6}+6   \right) \nonumber \\
	&= \left( \begin{array}{ccc|ccc|ccc|ccc}
	\star & 0&1&\star&\star&\star & \star & 6&7&\star&\star&\star \\
	\star & 3&4&\star&\star&\star &  \star & 9&10&\star&\star&\star\\
	\star & \star&2&0&\star&\star & 	 \star & \star&8&6&\star&\star \\
	\star & \star&5&3&\star&\star & 	 \star & \star&11&9&\star&\star \\
	\star & \star&\star&1&2&\star & 	 \star & \star&\star&7&8&\star \\
	\star & \star&\star&4&5&\star & 	 \star & \star&\star&10&11&\star \\ \hline
	\star & \star&\star&\star&0&1 & 	 \star & \star&\star&\star&6&7 \\
	\star & \star&\star&\star&3&4 &	 \star & \star&\star&\star&9&10 \\
	0 & \star&\star&\star&\star&2 & 	 6 & \star&\star&\star&\star&8 \\
	3 & \star&\star&\star&\star&5 &	 9 & \star&\star&\star&\star&11 \\
	1 & 2&\star&\star&\star&\star &	 7 & 8&\star&\star&\star&\star \\
	4 & 5&\star&\star&\star&\star & 	 10 & 11&\star&\star&\star&\star 
	\end{array} \right ) \label{exmp final matrix }
	\end{align}
	There are $12$ rows, $0,1,2,\ldots,11$, and $12$ columns, $0,1,2,\ldots,11$, in the PDA $\textbf{P}_{12 \times 12}$.  The rows correspond to the sub-files of each file and columns correspond to the users. That is, the $i^{th}$ row, $i \in [0,12)$, represents the sub-files $W_{n,i}, \forall n \in [0,12)$, while the  $j^{th}$ column, $j \in [0,12)$, represents the user $U_{j}$. 
	There are a total of $12$ integers in $\textbf{P}_{12 \times 12}$, and each integer occurs $4$ times in $\textbf{P}_{12 \times 12}$. All the conditions $\textit{C1},\textit{C2}'$ and \textit{C3} are satisfied by the array $\textbf{P}_{12 \times 12}$.  Additionally, all the $8$ stars in each column occur in consecutive rows and the position of stars in each column can be obtained by cyclically shifting the previous column by $2$ units down. Therefore, the array $\textbf{P}_{12 \times 12}$ represents a $2$-cyclic $4$-regular $(12,12,8,12)$ PDA. 
	
	Let the demand vector be ${\bf d}=(d_0,d_1,d_2,d_3,d_4,d_5,d_6,d_7,d_8,d_9,d_{10},d_{11})$. 
	We obtain a delivery scheme for the multi-access coded caching scheme with $K=N=12,k=2,L=4$ and sub-packetization level $F=12$, using the caching scheme given by Algorithm \ref{algo1} based on the PDA $ \textbf{P}_{12 \times 12}$. The following are the coded symbols obtained using Algorithm \ref{algo1}:
	
	\begin{align*}
	\mathcal{Y}_0 &= W_{d_0,8} \oplus W_{d_1,0} \oplus W_{d_3,2} \oplus W_{d_4,6}\\
	\mathcal{Y}_1 &= W_{d_0,10} \oplus W_{d_2,0} \oplus W_{d_3,4} \oplus W_{d_5,6}\\		
	\mathcal{Y}_2 &= W_{d_1,10} \oplus W_{d_2,2} \oplus W_{d_4,4} \oplus W_{d_5,8}\\	
	\mathcal{Y}_3 &= W_{d_0,9} \oplus W_{d_1,1} \oplus W_{d_3,3} \oplus W_{d_4,7}\\		
	\mathcal{Y}_4 &= W_{d_0,11} \oplus W_{d_2,1} \oplus W_{d_3,5} \oplus W_{d_5,7}\\	
	\mathcal{Y}_5 &= W_{d_1,11} \oplus W_{d_2,3} \oplus W_{d_4,5} \oplus W_{d_5,9}\\		
	\mathcal{Y}_6 &= W_{d_6,8} \oplus W_{d_7,0} \oplus W_{d_9,2} \oplus W_{d_{10},6}\\
	\mathcal{Y}_7 &= W_{d_6,10} \oplus W_{d_8,0} \oplus W_{d_9,4} \oplus W_{d_{11},6}\\
	\mathcal{Y}_8 &= W_{d_7,10} \oplus W_{d_8,2} \oplus W_{d_{10},4} \oplus W_{d_{11},8}\\
	\mathcal{Y}_9 &= W_{d_6,9} \oplus W_{d_7,1} \oplus W_{d_9,3} \oplus W_{d_{10},7}\\		
	\mathcal{Y}_{10} &= W_{d_6,11} \oplus W_{d_8,1} \oplus W_{d_9,5} \oplus W_{d_{11},7}\\		
	\mathcal{Y}_{11} &= W_{d_7,11} \oplus W_{d_8,3} \oplus W_{d_{10},5} \oplus W_{d_{11},9}.			
	\end{align*}
	Hence, the PDA $\textbf{P}_{12 \times 12}$ characterizes the delivery scheme for this example.
	For the caching scheme generated by a $2$-cyclic $4$-regular $(12,12,8,12)$ PDA using Algorithm \ref{algo1}, the coding gain equals to $4,$ since each transmission benefits $4$ users. The coding gain achieved using the NT and RK schemes for this example is $3$ while it is $4$ for SPE  scheme. The sub-packetization level required for NT, RK and SPE schemes are $180,30$ and $18$  respectively. The transmission rate achieved using our scheme is $1$ while the transmission rate achieved using the NK scheme is $0.755$. However the sub-packetization level required for the NK scheme is $30$
	while for our scheme the sub-packetization level required is $12$ which is less than half of that value.	
\end{example}

\subsection{Delivery Scheme}
In this section we provide the delivery scheme to prove Theorem \ref{thm1}.
The parts of the file $W_{{d_{\alpha}}}$ available with the user $U_{\alpha}, \alpha \in [0,K)$, are $kL$ consecutive sub-files, i.e.,  $\{W_{d_{\alpha},(k \alpha+i) \text{ mod } K} : i \in [0,kL)\}$. Hence,  the user $U_{\alpha}$ should be able to decode all the remaining $K-kL$ sub-files, i.e.,  $\{W_{d_{\alpha},(k \alpha+kL+i) \text{ mod } K} : i \in [0,K-k L)\}$. To retrieve the remaining sub-files, the transmitted symbols are obtained using the PDA constructed in Algorithm \ref{algo2} for the case discussed in Theorem \ref{thm1}.
\begin{thm} \label{thm valid pda}
	The $K \times K$ matrix ${\bf P}$ obtained using Algorithm \ref{algo2} is a $k$-cyclic $\frac{2K}{K-kL+k}$-regular $(K,K,kL,\frac{(K-kL)(K-kL+k)}{2} )$ PDA.
\end{thm}
	The proof of Theorem \ref{thm valid pda} is provided in Section IV.
\begin{thm} \label{thm delivery proof}
	For a given $k$-cyclic $\frac{2K}{K-kL+k}$-regular $(K,K,kL,\frac{(K-kL)(K-kL+k)}{2})$ PDA, $\textbf{P} = (p_{i,j}), i \in [0, K ), j \in [0, K )$ constructed using Algorithm \ref{algo2}, we can obtain a delivery algorithm using Algorithm \ref{algo1} for the considered case of multi-access coded problem in Theorem \ref{thm1}, with sub-packetization level $F=K$. For any demand vector $\textbf{d},$ the demands of all the users are met at the transmission rate of $\mathcal{R}_{new}=\frac{(K-kL)(K-kL+k)}{2K}$.
\end{thm}
\begin{proof}
	The proof of this theorem is straightforward from Theorem \ref{thm PDA delivery}. In a $t$-cyclic $g$-regular $(K,F,Z,S)$ PDA \textbf{P}, rows correspond to the sub-files of each file and columns correspond to the users. In any column $j \in [0,K)$, if $ p_{i,j} = \star,$ then user $U_j$ has access to the $i^{th}$ sub-file of all the files. If $ p_{i,j} = s$ is an integer, then it implies the $i^{th}$ sub-file of none of the files is accessible to the user $U_j$. The condition \textit{C1} implies some $Z$ sub-files of all the files are accessible to each user. The condition \textit{C2} implies that the number of symbols transmitted by the server is exactly $S$ since XOR of the requested packets indicated by $s$ is broadcast by the server during the slot for $s$. Hence the transmission rate is $\frac{S}{F}$. The condition \textit{C3} makes sure that each user can get the demanded sub-file, since all the other sub-files in the coded symbol are available at its cache.
	The three conditions in the Defintion \ref{def: PDA} of the PDA guarantees all the users retrieve the requested files.
\end{proof}
\begin{algorithm}
	\caption{$k$-cyclic $\frac{2K}{K-kL+k}$-regular $(K,K,kL,\frac{(K-kL)(K-kL+k)}{2})$ PDA {\bf P} Construction, where  $\gamma \in \left \{\frac{k}{K}: k | K, (K-kL+k)|K, k\in \left [1,K \right ]\right \}$}
	\label{algo2}
	\begin{algorithmic}[1]
		\Procedure{{\bf 1}: Construct a square matrix $ \textbf{A} =(a_{i,j}), i , j \in \left[0,\frac{K-kL}{k} \right ].$}{}
	\For{\textit{$i \in \left [0,\frac{K-kL}{k} \right ]$}}
\For{\textit{$j \in \left [0,\frac{K-kL}{k} \right ]$}}
\If{$j\leq i$}
\State $a_{i,j} =\star$
\ElsIf{$i=0$ and $j>i$}
\State $a_{i,j} =j-1$
\Else
\State $a_{i,j} =a_{i-1,j} + \left (\frac{K-kL}{k}-i \right )$
\EndIf
\EndFor
\EndFor
\EndProcedure \textbf{ 1}
\Procedure{{\bf 2}: Obtain a square matrix $ {\bf P}_1 =(p_{i,j}),  i,j \in [0,\frac{K}{k})$.}{}
{\small$$
\textbf{P}_1 = \left ( \begin{array}{ccccccc}
{\bf A}&{\bf A}^{T} & {\bf X}&{\bf X}&\ldots &{\bf X}&{\bf X}\\
{\bf X}&{\bf A}&{\bf A}^{T} &{\bf X}&\ldots & {\bf X}&{\bf X}\\
{\bf X}&{\bf X}&{\bf A}&{\bf A}^{T} &\ldots & {\bf X}&{\bf X}\\
\vdots &\vdots &\vdots &\vdots &\vdots &\ldots &\vdots \\
{\bf X}&{\bf X}&  {\bf X}&{\bf X}&\ldots &{\bf A}&{\bf A}^{T}\\
{\bf A}^T&{\bf X}&{\bf X}&{\bf X} &\ldots & {\bf X}&{\bf A}\\
\end{array}\right )$$}
where $ {\bf X}$ represents a $\frac{K-kL+k}{k} \times \frac{K-kL+k}{k}$ matrix with all the entries being $\star$. The matrix $ {\bf P}_1$ is a block matrix having $\frac{K}{K-kL+k}$ row and column blocks with blocks of size $\frac{K-kL+k}{k} \times \frac{K-kL+k}{k}$ as entries.
\EndProcedure \textbf{ 2}
\Procedure{{\bf 3}: From the matrix ${\bf P}_1$, generate a  $K \times \frac{K}{k}$  matrix $\tilde{\textbf{P}}_1 = (\tilde{p}_{i,j}) , i \in [0,K),j\in [0,\frac{K}{k}).$}{}
	\State Let $S_1=\left ( \frac{(K-kL)(K-kL+k)}{2k^2} \right ).$
\For{\textit{$i \in \left [0,K \right )$}}
\For{\textit{$j \in \left [0,\frac{K}{k} \right )$}}
\If{$k|i$}
\State $\tilde{p}_{i,j} =p_{(\frac{i}{k}),j}. $
\Else
\State $\tilde{p}_{i,j} =\tilde{p}_{i-1,j}  + S_1,$ where $S_1 +\star =\star.$
\EndIf
\EndFor
\EndFor
\EndProcedure \textbf{ 3}
\Procedure{{\bf 4}: Now, obtain a $K \times K$ matrix $\textbf{P}$}{}
$$
\textbf{P} =  \left ( \begin{array}{cccccc}
\tilde{\textbf{P}}_1 & \tilde{\textbf{P}}_2 & \tilde{\textbf{P}}_3 &\ldots &\tilde{\textbf{P}}_k
\end{array}\right )$$
where each matrix $ \tilde{\textbf{P}}_t =\tilde{\textbf{P}}_1+(t-1)\tilde{S}_1,t \in  [2,k],  \tilde{S}_1=\frac{(K-kL)(K-kL+k)}{2k},$ is a $K \times \frac{K}{t}$ matrix, with $ \star+(t-1)\tilde{S}_1 =\star$.
\EndProcedure \textbf{ 4}
	\end{algorithmic}
\end{algorithm}
\begin{example}
	Consider Example \ref{exmp macc} with $K=N=12,k=2,L=4.$ 
	We illustrate each procedure involved in obtaining the matrix $\textbf{P}_{12 \times 12}$ using Algorithm \ref{algo2} for Example \ref{exmp macc}.
	The ${\bf A}$ matrix obtained by \textbf{procedure 1} in Algorithm \ref{algo2} corresponding to this example is $${\bf A} = \left( \begin{array}{ccccc}
	\star &0 & 1\\
	\star & \star & 2\\
	\star & \star & \star
	\end{array}\right).$$
	The ${\bf P}_1$ matrix obtained by \textbf{procedure 2} in Algorithm \ref{algo2} corresponding to this example is $${\bf P}_1 =\left ( \begin{array}{cccccc}
{\bf A} &{\bf A}^{T} \\
{\bf A}^{T} & {\bf A}
	\end{array} \right ) = \left( \begin{array}{ccc|ccccc}
	\star &0 & 1& \star& \star &\star \\
	\star & \star & 2 &0 &\star &\star\\
		\star & \star & \star& 1&2&\star \\ \hline
		\star &\star & \star &\star &0&1\\
		0&\star&\star &\star&\star &2\\
		1&2&\star &\star &\star &\star
	\end{array}\right).$$
	Using \textbf{procedure 3} in Algorithm \ref{algo2}, we obtain the matrix $\tilde{{\bf P}}_1$ as in (\ref{PDA 12x6}).
	
	Now, using \textbf{procedure 4}, the matrix $\tilde{{\bf P}}_2$ obtained is as follows: $$\tilde{{\bf P}}_2=\tilde{{\bf P}}_1+6=\left( \begin{array}{ccc|ccccc}
	\star &6 & 7& \star& \star &\star \\
	\textcolor{red}{\star} &\textcolor{red}{9} & \textcolor{red}{10}& \textcolor{red}{\star}& \textcolor{red}{\star} &\textcolor{red}{\star} \\ 
	\star & \star & 8 &6 &\star &\star\\
	\textcolor{red}{\star} & \textcolor{red}{\star} & \textcolor{red}{11} &\textcolor{red}{9} &\textcolor{red}{\star} &\textcolor{red}{\star}\\ 
	\star & \star & \star& 7&8&\star \\ 
	\textcolor{red}{\star} & \textcolor{red}{\star} & \textcolor{red}{\star}& \textcolor{red}{10}&\textcolor{red}{11}&\textcolor{red}{\star} \\  \hline
	\star &\star & \star &\star &6&7\\
	\textcolor{red}{\star} &\textcolor{red}{\star} & \textcolor{red}{\star} &\textcolor{red}{\star} &\textcolor{red}{9}&\textcolor{red}{10}\\ 
	6&\star&\star &\star&\star &8\\
	\textcolor{red}{9}&\textcolor{red}{\star}&\textcolor{red}{\star} &\textcolor{red}{\star}&\textcolor{red}{\star}& \textcolor{red}{11}\\
	7&8&\star &\star &\star &\star\\
	\textcolor{red}{10}&\textcolor{red}{11}&\textcolor{red}{\star} &\textcolor{red}{\star} &\textcolor{red}{\star} &\textcolor{red}{\star}
	\end{array}\right)$$ 
	Finally, the matrix ${\bf P}$ obtained by concatenating $\tilde{{\bf P}}_1$ and $\tilde{{\bf P}}_2$ is as in (\ref{exmp final matrix }).
\end{example}
\begin{example} \label{exmp macc 2}
	Consider another example with $N=K=36,k=3,L=9$. The server has access to $36$ files: $\mathcal{W}=\{W_n : n \in [0,36)\}$, and each file $W_n,n\in [0,36)$, is divided into $36$ sub-files: $W_n=\{W_{n,j} : j \in [0,36)\}$. Each cache $C_{\alpha}, \alpha \in [0,36),$ is filled as $M_{\alpha}= \{W_{n,(3 \alpha +j) \text{ mod } 36}:  j \in \{0,1,2\},   n \in [0,36)\}$.
	
	Each user $U_{\alpha}, \alpha \in [0,36)$, has access to all the caches in the set $\{C_{(\alpha+i) \text{ mod } 36} : i \in [0,9)\}$.
	Hence, for each user $U_{\alpha},$ the accessible cache content is $\mathcal{Z}_{\alpha} =\{W_{n,(3\alpha + i) \text{ mod } 36} : i \in [0,27),n \in [0,36)\}$.
	
		The ${\bf A}$ matrix obtained by \textbf{procedure 1} in Algorithm \ref{algo2} corresponding to this example is $${\bf A} = \left( \begin{array}{ccccc}
	\star &0 & 1&2\\
	\star & \star & 3&4\\
	\star & \star & \star&5\\
	\star&\star &\star&\star
	\end{array}\right).$$
	The ${\bf P}_1$ matrix obtained by \textbf{procedure 2} in Algorithm \ref{algo2} corresponding to this example is \begin{align}{\bf P}_1 &=\left ( \begin{array}{cccccc}
	{\bf A} &{\bf A}^{T}& \textbf{X} \\
	 \textbf{X}&{\bf A} &{\bf A}^{T} \\
	{\bf A}^{T} & \textbf{X}& {\bf A}
	\end{array} \right ) \\&= \left( \begin{array}{cccc|cccc|cccc}
	\star &0 & 1& 2&\star& \star &\star&\star&\star&\star&\star&\star \\
	\star & \star & 3 &4&0 &\star &\star &\star &\star&\star&\star&\star\\
	\star & \star & \star& 5&1&3&\star &\star&\star&\star&\star&\star\\ 
	\star &\star & \star &\star &2&4 &5 &\star &\star&\star&\star&\star\\ \hline
	\star&\star&\star&\star&\star&0&1&2&\star&\star&\star&\star\\
	\star&\star&\star&\star&\star&\star&3&4&0&\star&\star&\star\\	\star&\star&\star&\star&	\star & \star & \star& 5&1&3&\star &\star\\ \star&\star&\star&\star&	\star &\star & \star &\star &2&4 &5 &\star \\ \hline
	\star&\star&\star&\star&	\star&\star&\star&\star&\star&0&1&2\\
	0&\star&\star&\star&\star&\star&\star&\star&\star&\star&3&4\\1&3&\star &\star&	\star&\star&\star&\star&	\star & \star & \star& 5\\ 2&4 &5 &\star &\star&\star&\star&\star&	\star &\star & \star &\star \\ 
	\end{array}\right).
	\end{align}
	Using \textbf{procedure 3} in Algorithm \ref{algo2}, we obtain the matrix $\tilde{{\bf P}}_1$ as given below. $\tilde{{\bf P}}_1 =$
	{\small
		\begin{align}  \left( \begin{array}{cccc|cccc|cccc}
		\star &0 & 1& 2&\star& \star &\star&\star&\star&\star&\star&\star \\
		\textcolor{red}{\star} &	\textcolor{red}{6} & 	\textcolor{red}{7}& 	\textcolor{red}{8}&	\textcolor{red}{\star}&	\textcolor{red}{ \star} &	\textcolor{red}{\star}&	\textcolor{red}{\star}&	\textcolor{red}{\star}&	\textcolor{red}{\star}&	\textcolor{red}{\star}&	\textcolor{red}{\star} \\
		\textcolor{blue}{\star} &\textcolor{blue}{12} & \textcolor{blue}{13}& \textcolor{blue}{14}&\textcolor{blue}{\star}&\textcolor{blue}{ \star} &\textcolor{blue}{\star}&\textcolor{blue}{\star}&\textcolor{blue}{\star}&\textcolor{blue}{\star}&\textcolor{blue}{\star}&\textcolor{blue}{\star} \\
		\star & \star & 3 &4&0 &\star &\star &\star &\star&\star&\star&\star\\
		\textcolor{red}{\star} &	\textcolor{red}{ \star} &	\textcolor{red}{ 9} &	\textcolor{red}{10}&	\textcolor{red}{6} &	\textcolor{red}{\star} &	\textcolor{red}{\star} &	\textcolor{red}{\star} &	\textcolor{red}{\star}&	\textcolor{red}{\star}&	\textcolor{red}{\star}&	\textcolor{red}{\star}\\
		\textcolor{blue}{\star} & \textcolor{blue}{\star} & \textcolor{blue}{15} &\textcolor{blue}{16}&\textcolor{blue}{12 }&\textcolor{blue}{\star} &\textcolor{blue}{\star }&\textcolor{blue}{\star} &\textcolor{blue}{\star}&\textcolor{blue}{\star}&\textcolor{blue}{\star}&\textcolor{blue}{\star}\\
		\star & \star & \star& 5&1&3&\star &\star&\star&\star&\star&\star\\ 
		\textcolor{red}{	\star} &	\textcolor{red}{ \star} &	\textcolor{red}{ \star}& 	\textcolor{red}{11}&	\textcolor{red}{7}&	\textcolor{red}{9}&	\textcolor{red}{\star} &	\textcolor{red}{\star}&	\textcolor{red}{\star}&	\textcolor{red}{\star}&	\textcolor{red}{\star}&	\textcolor{red}{\star}\\ 
		\textcolor{blue}{\star} & \textcolor{blue}{\star} & \textcolor{blue}{\star}& \textcolor{blue}{17}&\textcolor{blue}{13}&\textcolor{blue}{15}&\textcolor{blue}{\star} &\textcolor{blue}{\star}&\textcolor{blue}{\star}&\textcolor{blue}{\star}&\textcolor{blue}{\star}&\textcolor{blue}{\star}\\ 
		\star &\star & \star &\star &2&4 &5 &\star &\star&\star&\star&\star\\ 
		\textcolor{red}{\star} &	\textcolor{red}{\star} &	\textcolor{red}{ \star} &	\textcolor{red}{\star} &	\textcolor{red}{8}&	\textcolor{red}{10} &	\textcolor{red}{11 }&	\textcolor{red}{\star} &	\textcolor{red}{\star}&	\textcolor{red}{\star}&	\textcolor{red}{\star}&	\textcolor{red}{\star}\\
		\textcolor{blue}{\star} &\textcolor{blue}{\star} & \textcolor{blue}{\star} &\textcolor{blue}{\star} &\textcolor{blue}{14}&\textcolor{blue}{16} &\textcolor{blue}{17} &\textcolor{blue}{\star} &\textcolor{blue}{\star}&\textcolor{blue}{\star}&\textcolor{blue}{\star}&\textcolor{blue}{\star}\\\hline
		\star&\star&\star&\star&\star&0&1&2&\star&\star&\star&\star\\
		\textcolor{red}{\star}&	\textcolor{red}{\star}&	\textcolor{red}{\star}&	\textcolor{red}{\star}&	\textcolor{red}{\star}&	\textcolor{red}{6}&	\textcolor{red}{7}&	\textcolor{red}{8}&	\textcolor{red}{\star}&	\textcolor{red}{\star}&	\textcolor{red}{\star}&	\textcolor{red}{\star}\\
		\textcolor{blue}{\star}&\textcolor{blue}{\star}&\textcolor{blue}{\star}&\textcolor{blue}{\star}&\textcolor{blue}{\star}&\textcolor{blue}{12}&\textcolor{blue}{13}&\textcolor{blue}{14}&\textcolor{blue}{\star}&\textcolor{blue}{\star}&\textcolor{blue}{\star}&\textcolor{blue}{\star}\\
		\star&\star&\star&\star&\star&\star&3&4&0&\star&\star&\star\\
		\textcolor{red}{\star}&	\textcolor{red}{\star}&	\textcolor{red}{\star}&	\textcolor{red}{\star}&	\textcolor{red}{\star}&	\textcolor{red}{\star}&	\textcolor{red}{9}&	\textcolor{red}{10}&	\textcolor{red}{6}&	\textcolor{red}{\star}&	\textcolor{red}{\star}&	\textcolor{red}{\star}\\
		\textcolor{blue}{\star}&\textcolor{blue}{\star}&\textcolor{blue}{\star}&\textcolor{blue}{\star}&\textcolor{blue}{\star}&\textcolor{blue}{\star}&\textcolor{blue}{15}&\textcolor{blue}{16}&\textcolor{blue}{12}&\textcolor{blue}{\star}&\textcolor{blue}{\star}&\textcolor{blue}{\star}\\	
		\star&\star&\star&\star&	\star & \star & \star& 5&1&3&\star &\star\\ 
		\textcolor{red}{\star}&	\textcolor{red}{\star}&	\textcolor{red}{\star}&	\textcolor{red}{\star}&	\textcolor{red}{	\star} &	\textcolor{red}{ \star} & 	\textcolor{red}{\star}&	\textcolor{red}{ 11}&	\textcolor{red}{7}&	\textcolor{red}{9}&	\textcolor{red}{\star} &	\textcolor{red}{\star}\\ 
		\textcolor{blue}{\star}&\textcolor{blue}{\star}&\textcolor{blue}{\star}&\textcolor{blue}{\star}&\textcolor{blue}{	\star} & \textcolor{blue}{\star} & \textcolor{blue}{\star}&\textcolor{blue}{ 17}&\textcolor{blue}{13}&\textcolor{blue}{15}&\textcolor{blue}{\star} &\textcolor{blue}{\star}\\ 
		\star&\star&\star&\star&	\star &\star & \star &\star &2&4 &5 &\star \\ 
		\textcolor{red}{\star}&	\textcolor{red}{\star}&	\textcolor{red}{\star}&	\textcolor{red}{\star}&	\textcolor{red}{	\star} &	\textcolor{red}{\star} &	\textcolor{red}{ \star}&	\textcolor{red}{\star} &	\textcolor{red}{8}&	\textcolor{red}{10} &	\textcolor{red}{11 }&	\textcolor{red}{\star} \\
		\textcolor{blue}{\star}&\textcolor{blue}{\star}&\textcolor{blue}{\star}&\textcolor{blue}{\star}&	\textcolor{blue}{\star} &\textcolor{blue}{\star} & \textcolor{blue}{\star} &\textcolor{blue}{\star} &\textcolor{blue}{14}&\textcolor{blue}{16} &\textcolor{blue}{17} &\textcolor{blue}{\star} \\\hline
		\star&\star&\star&\star&	\star&\star&\star&\star&\star&0&1&2\\
		\textcolor{red}{	\star}&	\textcolor{red}{\star}&	\textcolor{red}{\star}&	\textcolor{red}{\star}&		\textcolor{red}{\star}&	\textcolor{red}{\star}&	\textcolor{red}{\star}&	\textcolor{red}{\star}&	\textcolor{red}{\star}&	\textcolor{red}{6}&	\textcolor{red}{7}&	\textcolor{red}{8}\\
		\textcolor{blue}{\star}&\textcolor{blue}{\star}&\textcolor{blue}{\star}&\textcolor{blue}{\star}&	\textcolor{blue}{\star}&\textcolor{blue}{\star}&\textcolor{blue}{\star}&\textcolor{blue}{\star}&\textcolor{blue}{\star}&\textcolor{blue}{12}&\textcolor{blue}{13}&\textcolor{blue}{14}\\
		0&\star&\star&\star&\star&\star&\star&\star&\star&\star&3&4\\
		\textcolor{red}{6}&	\textcolor{red}{\star}&	\textcolor{red}{\star}&	\textcolor{red}{\star}&	\textcolor{red}{\star}&	\textcolor{red}{\star}&	\textcolor{red}{\star}&	\textcolor{red}{\star}&	\textcolor{red}{\star}&	\textcolor{red}{\star}&	\textcolor{red}{9}&	\textcolor{red}{10}\\
		\textcolor{blue}{12}&\textcolor{blue}{\star}&\textcolor{blue}{\star}&\textcolor{blue}{\star}&\textcolor{blue}{\star}&\textcolor{blue}{\star}&\textcolor{blue}{\star}&\textcolor{blue}{\star}&\textcolor{blue}{\star}&\textcolor{blue}{\star}&\textcolor{blue}{15}&\textcolor{blue}{16}\\
		1&3&\star &\star&	\star&\star&\star&\star&	\star & \star & \star& 5\\ 
		\textcolor{red}{7}&	\textcolor{red}{9}&	\textcolor{red}{\star} &	\textcolor{red}{\star}&		\textcolor{red}{\star}&	\textcolor{red}{\star}&	\textcolor{red}{\star}&	\textcolor{red}{\star}&		\textcolor{red}{\star} & 	\textcolor{red}{\star} & 	\textcolor{red}{\star}&	\textcolor{red}{ 11}\\ 
		\textcolor{blue}{13}&\textcolor{blue}{15}&\textcolor{blue}{\star} &\textcolor{blue}{\star}&\textcolor{blue}{	\star}&\textcolor{blue}{\star}&\textcolor{blue}{\star}&\textcolor{blue}{\star}&\textcolor{blue}{	\star} &\textcolor{blue}{ \star} & \textcolor{blue}{\star}&\textcolor{blue}{ 17}\\ 
		2&4 &5 &\star &\star&\star&\star&\star&	\star &\star & \star &\star \\
		\textcolor{red}{8}&	\textcolor{red}{10} &	\textcolor{red}{11} &	\textcolor{red}{\star} &	\textcolor{red}{\star}&	\textcolor{red}{\star}&	\textcolor{red}{\star}&	\textcolor{red}{\star}&	\textcolor{red}{	\star} &	\textcolor{red}{\star }& 	\textcolor{red}{\star} &	\textcolor{red}{\star} \\
		\textcolor{blue}{14}& \textcolor{blue}{16}&\textcolor{blue}{17} &\textcolor{blue}{\star} &\textcolor{blue}{\star}&\textcolor{blue}{\star}&\textcolor{blue}{\star}&\textcolor{blue}{\star}&\textcolor{blue}{	\star} &\textcolor{blue}{\star} & \textcolor{blue}{\star} &\textcolor{blue}{\star} 
		\end{array}\right). 
		\end{align}}

	Now, using \textbf{procedure 4}, the matrix $\tilde{{\bf P}}_2$ obtained is as follows: $\tilde{{\bf P}}_2 =\tilde{{\bf P}}_1+18=$
	{\small
		\begin{align}  \left( \begin{array}{cccc|cccc|cccc}
		\star &18 & 19& 20&\star& \star &\star&\star&\star&\star&\star&\star \\
		\textcolor{red}{\star} &	\textcolor{red}{24} & 	\textcolor{red}{25}& 	\textcolor{red}{26}&	\textcolor{red}{\star}&	\textcolor{red}{ \star} &	\textcolor{red}{\star}&	\textcolor{red}{\star}&	\textcolor{red}{\star}&	\textcolor{red}{\star}&	\textcolor{red}{\star}&	\textcolor{red}{\star} \\
		\textcolor{blue}{\star} &\textcolor{blue}{30} & \textcolor{blue}{31}& \textcolor{blue}{32}&\textcolor{blue}{\star}&\textcolor{blue}{ \star} &\textcolor{blue}{\star}&\textcolor{blue}{\star}&\textcolor{blue}{\star}&\textcolor{blue}{\star}&\textcolor{blue}{\star}&\textcolor{blue}{\star} \\
		\star & \star & 21 &22&18 &\star &\star &\star &\star&\star&\star&\star\\
			\textcolor{red}{\star} &	\textcolor{red}{ \star} &	\textcolor{red}{ 27} &	\textcolor{red}{28}&	\textcolor{red}{24} &	\textcolor{red}{\star} &	\textcolor{red}{\star} &	\textcolor{red}{\star} &	\textcolor{red}{\star}&	\textcolor{red}{\star}&	\textcolor{red}{\star}&	\textcolor{red}{\star}\\
		\textcolor{blue}{\star} & \textcolor{blue}{\star} & \textcolor{blue}{33} &\textcolor{blue}{34}&\textcolor{blue}{30 }&\textcolor{blue}{\star} &\textcolor{blue}{\star }&\textcolor{blue}{\star} &\textcolor{blue}{\star}&\textcolor{blue}{\star}&\textcolor{blue}{\star}&\textcolor{blue}{\star}\\
		\star & \star & \star& 23&19&21&\star &\star&\star&\star&\star&\star\\ 
		\textcolor{red}{	\star} &	\textcolor{red}{ \star} &	\textcolor{red}{ \star}& 	\textcolor{red}{29}&	\textcolor{red}{25}&	\textcolor{red}{27}&	\textcolor{red}{\star} &	\textcolor{red}{\star}&	\textcolor{red}{\star}&	\textcolor{red}{\star}&	\textcolor{red}{\star}&	\textcolor{red}{\star}\\ 
		\textcolor{blue}{\star} & \textcolor{blue}{\star} & \textcolor{blue}{\star}& \textcolor{blue}{35}&\textcolor{blue}{31}&\textcolor{blue}{33}&\textcolor{blue}{\star} &\textcolor{blue}{\star}&\textcolor{blue}{\star}&\textcolor{blue}{\star}&\textcolor{blue}{\star}&\textcolor{blue}{\star}\\ 
		\star &\star & \star &\star &20&22 &23 &\star &\star&\star&\star&\star\\ 
			\textcolor{red}{\star} &	\textcolor{red}{\star} &	\textcolor{red}{ \star} &	\textcolor{red}{\star} &	\textcolor{red}{26}&	\textcolor{red}{28} &	\textcolor{red}{29 }&	\textcolor{red}{\star} &	\textcolor{red}{\star}&	\textcolor{red}{\star}&	\textcolor{red}{\star}&	\textcolor{red}{\star}\\
		\textcolor{blue}{\star} &\textcolor{blue}{\star} & \textcolor{blue}{\star} &\textcolor{blue}{\star} &\textcolor{blue}{32}&\textcolor{blue}{34} &\textcolor{blue}{35} &\textcolor{blue}{\star} &\textcolor{blue}{\star}&\textcolor{blue}{\star}&\textcolor{blue}{\star}&\textcolor{blue}{\star}\\\hline
		\star&\star&\star&\star&\star&18&19&20&\star&\star&\star&\star\\
			\textcolor{red}{\star}&	\textcolor{red}{\star}&	\textcolor{red}{\star}&	\textcolor{red}{\star}&	\textcolor{red}{\star}&	\textcolor{red}{24}&	\textcolor{red}{25}&	\textcolor{red}{26}&	\textcolor{red}{\star}&	\textcolor{red}{\star}&	\textcolor{red}{\star}&	\textcolor{red}{\star}\\
		\textcolor{blue}{\star}&\textcolor{blue}{\star}&\textcolor{blue}{\star}&\textcolor{blue}{\star}&\textcolor{blue}{\star}&\textcolor{blue}{30}&\textcolor{blue}{31}&\textcolor{blue}{32}&\textcolor{blue}{\star}&\textcolor{blue}{\star}&\textcolor{blue}{\star}&\textcolor{blue}{\star}\\
		\star&\star&\star&\star&\star&\star&21&22&18&\star&\star&\star\\
			\textcolor{red}{\star}&	\textcolor{red}{\star}&	\textcolor{red}{\star}&	\textcolor{red}{\star}&	\textcolor{red}{\star}&	\textcolor{red}{\star}&	\textcolor{red}{27}&	\textcolor{red}{28}&	\textcolor{red}{24}&	\textcolor{red}{\star}&	\textcolor{red}{\star}&	\textcolor{red}{\star}\\
		\textcolor{blue}{\star}&\textcolor{blue}{\star}&\textcolor{blue}{\star}&\textcolor{blue}{\star}&\textcolor{blue}{\star}&\textcolor{blue}{\star}&\textcolor{blue}{33}&\textcolor{blue}{34}&\textcolor{blue}{30}&\textcolor{blue}{\star}&\textcolor{blue}{\star}&\textcolor{blue}{\star}\\	
		\star&\star&\star&\star&	\star & \star & \star& 23&19&21&\star &\star\\ 
			\textcolor{red}{\star}&	\textcolor{red}{\star}&	\textcolor{red}{\star}&	\textcolor{red}{\star}&	\textcolor{red}{	\star} &	\textcolor{red}{ \star} & 	\textcolor{red}{\star}&	\textcolor{red}{ 29}&	\textcolor{red}{25}&	\textcolor{red}{27}&	\textcolor{red}{\star} &	\textcolor{red}{\star}\\ 
		\textcolor{blue}{\star}&\textcolor{blue}{\star}&\textcolor{blue}{\star}&\textcolor{blue}{\star}&\textcolor{blue}{	\star} & \textcolor{blue}{\star} & \textcolor{blue}{\star}&\textcolor{blue}{ 35}&\textcolor{blue}{31}&\textcolor{blue}{33}&\textcolor{blue}{\star} &\textcolor{blue}{\star}\\ 
		\star&\star&\star&\star&	\star &\star & \star &\star &20&22 &23 &\star \\ 
			\textcolor{red}{\star}&	\textcolor{red}{\star}&	\textcolor{red}{\star}&	\textcolor{red}{\star}&	\textcolor{red}{	\star} &	\textcolor{red}{\star} &	\textcolor{red}{ \star}&	\textcolor{red}{\star} &	\textcolor{red}{26}&	\textcolor{red}{28} &	\textcolor{red}{29 }&	\textcolor{red}{\star} \\
		\textcolor{blue}{\star}&\textcolor{blue}{\star}&\textcolor{blue}{\star}&\textcolor{blue}{\star}&	\textcolor{blue}{\star} &\textcolor{blue}{\star} & \textcolor{blue}{\star} &\textcolor{blue}{\star} &\textcolor{blue}{32}&\textcolor{blue}{34} &\textcolor{blue}{35} &\textcolor{blue}{\star} \\\hline
		\star&\star&\star&\star&	\star&\star&\star&\star&\star&18&19&20\\
		\textcolor{red}{	\star}&	\textcolor{red}{\star}&	\textcolor{red}{\star}&	\textcolor{red}{\star}&		\textcolor{red}{\star}&	\textcolor{red}{\star}&	\textcolor{red}{\star}&	\textcolor{red}{\star}&	\textcolor{red}{\star}&	\textcolor{red}{24}&	\textcolor{red}{25}&	\textcolor{red}{26}\\
		\textcolor{blue}{\star}&\textcolor{blue}{\star}&\textcolor{blue}{\star}&\textcolor{blue}{\star}&	\textcolor{blue}{\star}&\textcolor{blue}{\star}&\textcolor{blue}{\star}&\textcolor{blue}{\star}&\textcolor{blue}{\star}&\textcolor{blue}{30}&\textcolor{blue}{31}&\textcolor{blue}{32}\\
		18&\star&\star&\star&\star&\star&\star&\star&\star&\star&21&22\\
			\textcolor{red}{24}&	\textcolor{red}{\star}&	\textcolor{red}{\star}&	\textcolor{red}{\star}&	\textcolor{red}{\star}&	\textcolor{red}{\star}&	\textcolor{red}{\star}&	\textcolor{red}{\star}&	\textcolor{red}{\star}&	\textcolor{red}{\star}&	\textcolor{red}{27}&	\textcolor{red}{28}\\
		\textcolor{blue}{30}&\textcolor{blue}{\star}&\textcolor{blue}{\star}&\textcolor{blue}{\star}&\textcolor{blue}{\star}&\textcolor{blue}{\star}&\textcolor{blue}{\star}&\textcolor{blue}{\star}&\textcolor{blue}{\star}&\textcolor{blue}{\star}&\textcolor{blue}{33}&\textcolor{blue}{34}\\
		19&21&\star &\star&	\star&\star&\star&\star&	\star & \star & \star& 23\\ 
			\textcolor{red}{25}&	\textcolor{red}{27}&	\textcolor{red}{\star} &	\textcolor{red}{\star}&		\textcolor{red}{\star}&	\textcolor{red}{\star}&	\textcolor{red}{\star}&	\textcolor{red}{\star}&		\textcolor{red}{\star} & 	\textcolor{red}{\star} & 	\textcolor{red}{\star}&	\textcolor{red}{ 29}\\ 
		\textcolor{blue}{31}&\textcolor{blue}{33}&\textcolor{blue}{\star} &\textcolor{blue}{\star}&\textcolor{blue}{	\star}&\textcolor{blue}{\star}&\textcolor{blue}{\star}&\textcolor{blue}{\star}&\textcolor{blue}{	\star} &\textcolor{blue}{ \star} & \textcolor{blue}{\star}&\textcolor{blue}{ 35}\\ 
		20&22 &23 &\star &\star&\star&\star&\star&	\star &\star & \star &\star \\
			\textcolor{red}{26}&	\textcolor{red}{28} &	\textcolor{red}{29} &	\textcolor{red}{\star} &	\textcolor{red}{\star}&	\textcolor{red}{\star}&	\textcolor{red}{\star}&	\textcolor{red}{\star}&	\textcolor{red}{	\star} &	\textcolor{red}{\star }& 	\textcolor{red}{\star} &	\textcolor{red}{\star} \\
		\textcolor{blue}{32}& \textcolor{blue}{34}&\textcolor{blue}{35} &\textcolor{blue}{\star} &\textcolor{blue}{\star}&\textcolor{blue}{\star}&\textcolor{blue}{\star}&\textcolor{blue}{\star}&\textcolor{blue}{	\star} &\textcolor{blue}{\star} & \textcolor{blue}{\star} &\textcolor{blue}{\star} 
		\end{array}\right). 
		\end{align}}
	 Similarly the matrix $\tilde{{\bf P}}_3$ obtained is as follows: $\tilde{{\bf P}}_3 =\tilde{{\bf P}}_1+36=$
	{\small
	\begin{align}  \left( \begin{array}{cccc|cccc|cccc}
	\star &36 & 37& 38&\star& \star &\star&\star&\star&\star&\star&\star \\
	\textcolor{red}{\star} &	\textcolor{red}{42} & 	\textcolor{red}{43}& 	\textcolor{red}{44}&	\textcolor{red}{\star}&	\textcolor{red}{ \star} &	\textcolor{red}{\star}&	\textcolor{red}{\star}&	\textcolor{red}{\star}&	\textcolor{red}{\star}&	\textcolor{red}{\star}&	\textcolor{red}{\star} \\
	\textcolor{blue}{\star} &\textcolor{blue}{48} & \textcolor{blue}{49}& \textcolor{blue}{50}&\textcolor{blue}{\star}&\textcolor{blue}{ \star} &\textcolor{blue}{\star}&\textcolor{blue}{\star}&\textcolor{blue}{\star}&\textcolor{blue}{\star}&\textcolor{blue}{\star}&\textcolor{blue}{\star} \\
	\star & \star & 39 &40&36 &\star &\star &\star &\star&\star&\star&\star\\
	\textcolor{red}{\star} &	\textcolor{red}{ \star} &	\textcolor{red}{ 45} &	\textcolor{red}{46}&	\textcolor{red}{42} &	\textcolor{red}{\star} &	\textcolor{red}{\star} &	\textcolor{red}{\star} &	\textcolor{red}{\star}&	\textcolor{red}{\star}&	\textcolor{red}{\star}&	\textcolor{red}{\star}\\
	\textcolor{blue}{\star} & \textcolor{blue}{\star} & \textcolor{blue}{51} &\textcolor{blue}{52}&\textcolor{blue}{48 }&\textcolor{blue}{\star} &\textcolor{blue}{\star }&\textcolor{blue}{\star} &\textcolor{blue}{\star}&\textcolor{blue}{\star}&\textcolor{blue}{\star}&\textcolor{blue}{\star}\\
	\star & \star & \star& 41&37&39&\star &\star&\star&\star&\star&\star\\ 
	\textcolor{red}{	\star} &	\textcolor{red}{ \star} &	\textcolor{red}{ \star}& 	\textcolor{red}{47}&	\textcolor{red}{43}&	\textcolor{red}{45}&	\textcolor{red}{\star} &	\textcolor{red}{\star}&	\textcolor{red}{\star}&	\textcolor{red}{\star}&	\textcolor{red}{\star}&	\textcolor{red}{\star}\\ 
	\textcolor{blue}{\star} & \textcolor{blue}{\star} & \textcolor{blue}{\star}& \textcolor{blue}{53}&\textcolor{blue}{49}&\textcolor{blue}{51}&\textcolor{blue}{\star} &\textcolor{blue}{\star}&\textcolor{blue}{\star}&\textcolor{blue}{\star}&\textcolor{blue}{\star}&\textcolor{blue}{\star}\\ 
	\star &\star & \star &\star &38&40 &41 &\star &\star&\star&\star&\star\\ 
	\textcolor{red}{\star} &	\textcolor{red}{\star} &	\textcolor{red}{ \star} &	\textcolor{red}{\star} &	\textcolor{red}{44}&	\textcolor{red}{46} &	\textcolor{red}{47 }&	\textcolor{red}{\star} &	\textcolor{red}{\star}&	\textcolor{red}{\star}&	\textcolor{red}{\star}&	\textcolor{red}{\star}\\
	\textcolor{blue}{\star} &\textcolor{blue}{\star} & \textcolor{blue}{\star} &\textcolor{blue}{\star} &\textcolor{blue}{50}&\textcolor{blue}{52} &\textcolor{blue}{53} &\textcolor{blue}{\star} &\textcolor{blue}{\star}&\textcolor{blue}{\star}&\textcolor{blue}{\star}&\textcolor{blue}{\star}\\\hline
	\star&\star&\star&\star&\star&36&37&38&\star&\star&\star&\star\\
	\textcolor{red}{\star}&	\textcolor{red}{\star}&	\textcolor{red}{\star}&	\textcolor{red}{\star}&	\textcolor{red}{\star}&	\textcolor{red}{42}&	\textcolor{red}{43}&	\textcolor{red}{44}&	\textcolor{red}{\star}&	\textcolor{red}{\star}&	\textcolor{red}{\star}&	\textcolor{red}{\star}\\
	\textcolor{blue}{\star}&\textcolor{blue}{\star}&\textcolor{blue}{\star}&\textcolor{blue}{\star}&\textcolor{blue}{\star}&\textcolor{blue}{48}&\textcolor{blue}{49}&\textcolor{blue}{50}&\textcolor{blue}{\star}&\textcolor{blue}{\star}&\textcolor{blue}{\star}&\textcolor{blue}{\star}\\
	\star&\star&\star&\star&\star&\star&39&40&36&\star&\star&\star\\
	\textcolor{red}{\star}&	\textcolor{red}{\star}&	\textcolor{red}{\star}&	\textcolor{red}{\star}&	\textcolor{red}{\star}&	\textcolor{red}{\star}&	\textcolor{red}{45}&	\textcolor{red}{46}&	\textcolor{red}{42}&	\textcolor{red}{\star}&	\textcolor{red}{\star}&	\textcolor{red}{\star}\\
	\textcolor{blue}{\star}&\textcolor{blue}{\star}&\textcolor{blue}{\star}&\textcolor{blue}{\star}&\textcolor{blue}{\star}&\textcolor{blue}{\star}&\textcolor{blue}{51}&\textcolor{blue}{52}&\textcolor{blue}{48}&\textcolor{blue}{\star}&\textcolor{blue}{\star}&\textcolor{blue}{\star}\\	
	\star&\star&\star&\star&	\star & \star & \star& 41&37&39&\star &\star\\ 
	\textcolor{red}{\star}&	\textcolor{red}{\star}&	\textcolor{red}{\star}&	\textcolor{red}{\star}&	\textcolor{red}{	\star} &	\textcolor{red}{ \star} & 	\textcolor{red}{\star}&	\textcolor{red}{ 47}&	\textcolor{red}{43}&	\textcolor{red}{45}&	\textcolor{red}{\star} &	\textcolor{red}{\star}\\ 
	\textcolor{blue}{\star}&\textcolor{blue}{\star}&\textcolor{blue}{\star}&\textcolor{blue}{\star}&\textcolor{blue}{	\star} & \textcolor{blue}{\star} & \textcolor{blue}{\star}&\textcolor{blue}{ 53}&\textcolor{blue}{49}&\textcolor{blue}{51}&\textcolor{blue}{\star} &\textcolor{blue}{\star}\\ 
	\star&\star&\star&\star&	\star &\star & \star &\star &38&40 &41 &\star \\ 
	\textcolor{red}{\star}&	\textcolor{red}{\star}&	\textcolor{red}{\star}&	\textcolor{red}{\star}&	\textcolor{red}{	\star} &	\textcolor{red}{\star} &	\textcolor{red}{ \star}&	\textcolor{red}{\star} &	\textcolor{red}{44}&	\textcolor{red}{46} &	\textcolor{red}{47 }&	\textcolor{red}{\star} \\
	\textcolor{blue}{\star}&\textcolor{blue}{\star}&\textcolor{blue}{\star}&\textcolor{blue}{\star}&	\textcolor{blue}{\star} &\textcolor{blue}{\star} & \textcolor{blue}{\star} &\textcolor{blue}{\star} &\textcolor{blue}{50}&\textcolor{blue}{52} &\textcolor{blue}{53} &\textcolor{blue}{\star} \\\hline
	\star&\star&\star&\star&	\star&\star&\star&\star&\star&36&37&38\\
	\textcolor{red}{	\star}&	\textcolor{red}{\star}&	\textcolor{red}{\star}&	\textcolor{red}{\star}&		\textcolor{red}{\star}&	\textcolor{red}{\star}&	\textcolor{red}{\star}&	\textcolor{red}{\star}&	\textcolor{red}{\star}&	\textcolor{red}{42}&	\textcolor{red}{43}&	\textcolor{red}{44}\\
	\textcolor{blue}{\star}&\textcolor{blue}{\star}&\textcolor{blue}{\star}&\textcolor{blue}{\star}&	\textcolor{blue}{\star}&\textcolor{blue}{\star}&\textcolor{blue}{\star}&\textcolor{blue}{\star}&\textcolor{blue}{\star}&\textcolor{blue}{48}&\textcolor{blue}{49}&\textcolor{blue}{50}\\
	36&\star&\star&\star&\star&\star&\star&\star&\star&\star&39&40\\
	\textcolor{red}{42}&	\textcolor{red}{\star}&	\textcolor{red}{\star}&	\textcolor{red}{\star}&	\textcolor{red}{\star}&	\textcolor{red}{\star}&	\textcolor{red}{\star}&	\textcolor{red}{\star}&	\textcolor{red}{\star}&	\textcolor{red}{\star}&	\textcolor{red}{45}&	\textcolor{red}{46}\\
	\textcolor{blue}{48}&\textcolor{blue}{\star}&\textcolor{blue}{\star}&\textcolor{blue}{\star}&\textcolor{blue}{\star}&\textcolor{blue}{\star}&\textcolor{blue}{\star}&\textcolor{blue}{\star}&\textcolor{blue}{\star}&\textcolor{blue}{\star}&\textcolor{blue}{51}&\textcolor{blue}{52}\\
	37&39&\star &\star&	\star&\star&\star&\star&	\star & \star & \star& 41\\ 
	\textcolor{red}{43}&	\textcolor{red}{45}&	\textcolor{red}{\star} &	\textcolor{red}{\star}&		\textcolor{red}{\star}&	\textcolor{red}{\star}&	\textcolor{red}{\star}&	\textcolor{red}{\star}&		\textcolor{red}{\star} & 	\textcolor{red}{\star} & 	\textcolor{red}{\star}&	\textcolor{red}{ 47}\\ 
	\textcolor{blue}{49}&\textcolor{blue}{51}&\textcolor{blue}{\star} &\textcolor{blue}{\star}&\textcolor{blue}{	\star}&\textcolor{blue}{\star}&\textcolor{blue}{\star}&\textcolor{blue}{\star}&\textcolor{blue}{	\star} &\textcolor{blue}{ \star} & \textcolor{blue}{\star}&\textcolor{blue}{ 53}\\ 
	38&40 &41 &\star &\star&\star&\star&\star&	\star &\star & \star &\star \\
	\textcolor{red}{44}&	\textcolor{red}{46} &	\textcolor{red}{47} &	\textcolor{red}{\star} &	\textcolor{red}{\star}&	\textcolor{red}{\star}&	\textcolor{red}{\star}&	\textcolor{red}{\star}&	\textcolor{red}{	\star} &	\textcolor{red}{\star }& 	\textcolor{red}{\star} &	\textcolor{red}{\star} \\
	\textcolor{blue}{50}& \textcolor{blue}{52}&\textcolor{blue}{53} &\textcolor{blue}{\star} &\textcolor{blue}{\star}&\textcolor{blue}{\star}&\textcolor{blue}{\star}&\textcolor{blue}{\star}&\textcolor{blue}{	\star} &\textcolor{blue}{\star} & \textcolor{blue}{\star} &\textcolor{blue}{\star} 
	\end{array}\right). 
	\end{align}}
	Finally, the matrix $\textbf{P}=\left (\begin{array}{ccc}
	\tilde{{\bf P}}_1&\tilde{{\bf P}}_2&\tilde{{\bf P}}_3
	\end{array} \right )$ is obtained by concatenating $\tilde{{\bf P}}_1,\tilde{{\bf P}}_2$ and $\tilde{{\bf P}}_3$.
	
	There are $36$ rows and columns in the matrix ${\bf P}$, indexed by  $0,1,2,\ldots,35$.  The rows correspond to the sub-files of each file and columns correspond to the users.
	There are a total of $54$ integers in $\textbf{P}$, and each integer occurs $6$ times in $\textbf{P}$. All the conditions $\textit{C1},\textit{C2}'$ and \textit{C3} are satisfied by the matrix $\textbf{P}$.  Additionally, all the $27$ stars in each column occur in consecutive rows and the position of stars in each column can be obtained by cyclically shifting the previous column by $3$ units down. Therefore, the matrix $\textbf{P}$ represents a $3$-cyclic $6$-regular $(36,36,27,54)$ PDA. 
	
	Using the caching scheme given by Algorithm \ref{algo1} based on the PDA $\textbf{P}$ constructed in this example, we obtain $54$ coded symbols. 
The coding gain achieved using this scheme is $6$, since each transmission benefits $6$ users. The coding gain achieved using the NT and RK schemes for this example is $4$. The sub-packetization level required for NT, and RK  schemes are $7920$ and $660$  respectively. 
\end{example}
\section{Proof of Theorem \ref{thm valid pda}} 
\label{proof of correctness}
In this section, we prove that the  matrix {\bf P} constructed using Algorithm \ref{algo2} is a $k$-cyclic $\frac{2K}{K-kL+k}$-regular $(K,K,kL,\frac{(K-kL)(K-kL+k)}{2})$ PDA.

Using \textbf{procedure 1} of Algorithm \ref{algo2}, we construct a square matrix {\bf A}. It can be observed from step 5 that all the diagonal entries as well as all the entries below the diagonal are stars. For $i=0$, we consider the $0^{th} $ row. The entry corresponding to the $0^{th}$ column is a star. From step 6, the entries corresponding to $1^{st}$ column to $(\frac{K-kL}{k})^{th}$ column are integers from $0$ to $\frac{K-kL}{k}-1$ respectively. Now,  for $i=1$, we consider the $1^{st} $ row and the entry corresponding to the $0^{th}$ and $1^{st}$ columns are stars. The entries corresponding to $2^{nd}$ column to $(\frac{K-kL}{k})^{th}$ column are integers from $\frac{K-kL}{k}$ to $\frac{2(K-kL)}{k}-2$ respectively. Similarly for any row $i \in [0,\frac{K-kL}{k}]$,  the entries corresponding to the $0^{th}$ column to the $i^{th}$ column are stars. The positions corresponding to the rest of the columns ($(i+1)^{th}$ column to $(\frac{K-kL}{k})^{th}$ column ) are filled with $\frac{K-kL}{k} -i$ different integers which are also different from the  integers used for filling the previous rows.
There are $S_1=\frac{(K-kL)(K-kL+k)}{2k^2}$ integers present in the matrix $\textbf{A}$, since there are  $S_1$ positions available above the diagonal and all those positions are filled with different integers. The number of stars in the $i^{th} $ column, $i \in [0,\frac{K-kL}{k}]$ of the matrix ${\bf A}$, is $\frac{K-kL+k}{k}-i$.

	Using \textbf{procedure 2} of Algorithm \ref{algo2}, we obtain a block matrix ${\bf P}_1$ with $\frac{K}{K-kL+k}$ row as well as column blocks, indexed by $[0,\frac{K}{K-kL+k})$. Each block is of size 
	$\frac{K-kL+k}{k} \times \frac{K-kL+k}{k}$. The $0^{th}$ column block is given by \begin{align} \label{P first column}
	\left ( \begin{array}{ccccccc}
	{\bf A}\\{\bf X}\\\vdots\\ {\bf X}\\{\bf A}^{T}
	\end{array}\right )
	\end{align} Each column block, $i \in [1,\frac{K}{K-kL+k})$ is obtained by cyclically shifting the blocks in the $0^{th}$ column block down by $i $ units. Each block in ${\bf P}_1$ is of size $\frac{K-kL+k}{k} \times \frac{K-kL+k}{k}$ and there are $\frac{K}{K-kL+k}$ rows and column blocks present in ${\bf P}_1$. Hence the matrix ${\bf P}_1$ is of size $\frac{K}{k} \times \frac{K}{k}.$
	
		Using \textbf{procedure 3} of Algorithm \ref{algo2}, a tall matrix $\tilde{{\bf P}}_1$ is generated from the square matrix ${\bf P}_1$. The $K \times \frac{K}{k}$ matrix $\tilde{{\bf P}}_1$ is generated from the $\frac{K}{k} \times \frac{K}{k}$ matrix ${\bf P}_1$ by adding new $k-1$ rows after each row in ${\bf P}_1$. 
	For each $i \in [0,\frac{K}{k})$, the $(ik)^{th}$ row of the matrix $\tilde{{\bf P}}_1$ is same as the $i^{th}$ row of the matrix ${\bf P}_1$ (as in step $21$).
	It is evident from step $23$ that the positions of stars in the next $k-1$ rows, i.e.,  $(ik+1)^{th}$ row to  $((i+1)k-1)^{th}$ row, are same as that of the $(ik)^{th}$ row of the matrix $\tilde{{\bf P}}_1$. 
	For the $(ik+j)^{th}$ row, $j \in [1,k)$ of the matrix $\tilde{{\bf P}}_1$ the entries corresponding to columns which do not have stars are filled with integer obtained by adding $S_1$ to the corresponding elements in the previous row of the matrix $\tilde{{\bf P}}_1$.  This is to make sure that the integers used for filling the $(ik+j)^{th}$ row is different from the integers used for filling the $(ik)^{th}$ row to $(ik+j-1)^{th}$ row. 

	Using the \textbf{procedure 4} of Algorithm \ref{algo2}, we define new matrices $\tilde{{\bf P}}_t, t \in [2,k]$, where $\tilde{{\bf P}}_t =(\tilde{p}_{i,j} + (t-1)\tilde{S}_1)$,  $\tilde{S}_1$ denotes the number of integers present in the matrix $\tilde{{\bf P}}_1$, and $ \star+(t-1)\tilde{S}_1  =\star$. Each integer entries in $ \tilde{{\bf P}}_t$ is obtained by adding $(t-1)\tilde{S}_1$ to the corresponding entry in $\tilde{{\bf P}}_1$. The matrix ${\bf P}$ is obtained by concatenating the $k$ matrices $\tilde{{\bf P}}_t, t \in [1,k]$, each of size $K \times \frac{K}{k}$. Hence the matrix ${\bf P}$ is of size $K \times K.$
	
	\begin{lem}
		The $\frac{K}{k} \times \frac{K}{k}$ matrix ${\bf P}_1$ obtained from \textbf{procedure 2} of Algorithm \ref{algo2} is a $1$-cyclic $\frac{2K}{K-kL+k}$-regular $(\frac{K}{k},\frac{K}{k},L,\frac{(K-kL)(K-kL+k)}{2k^2} )$ PDA.
	\end{lem}
	\begin{proof}
				The number of stars present  in the $i^{th} $ column, $i \in [0,\frac{K-kL}{k}]$ of the matrix ${\bf A}$, is $\frac{K-kL+k}{k}-i$ while the number of stars present  in the $i^{th} $ column of the matrix ${\bf A}^{T}$ (the transpose of the matrix ${\bf A}$) is $i+1$. Hence the number of stars present in the $i^{th} $ column, $i \in [0,\frac{K-kL}{k}]$ of the matrix ${\bf P}_1$, is $Z_1=(\frac{K-kL+k}{k}-i) + (\frac{K}{K-kL+k}-2)(\frac{K-kL+k}{k}) + (i+1) =L$. The number of stars in each column of the matrix ${\bf P}_1$ is also $L$, since each column block, $i \in [1,\frac{K}{K-kL+k})$ is obtained by cyclically shifting the blocks in the $0^{th}$ column block down by $i $ units, Hence the condition \textit{C1} in Definition \ref{def: PDA} is satisfied by the matrix ${\bf P}_1$. 
		
		The number of integers present in the matrix $\textbf{A}$ is $S_1 =\frac{(K-kL)(K-kL+k)}{2k^2}$. Each integer in the set $[0,S_1)$ occurs once in the matrix ${\bf A}$ as well as in ${\bf A}^{T}$. The block ${\bf A}$ as well as the  block ${\bf A}^{T}$ occur $\frac{K}{K-kL+k}$ times in the matrix ${\bf P}_1$ (once in each column block). The block ${\bf X}$ contains only stars. 
		Hence, 	the number of integers present in the matrix $\textbf{P}_1$ is $S_1$ and each integer in the set $[0,S_1)$ occurs $\frac{2K}{K-kL+k}$ times in the matrix ${\bf P}_1$. So, the condition $\textit{C2}'$ of Definition \ref{def regular PDA} is satisfied by the matrix ${\bf P}_1$ with $g=\frac{2K}{K-kL+k}$. 
		
	Now, we need to prove that the condition $\textit{C3}$ of Definition \ref{def: PDA} is satisfied by the matrix ${\bf P}_1$.
	Consider the sub-matrix $\tilde{{\bf A}} =(\tilde{a}_{i,j}), i \in [0,\frac{K-kL}{k}], j \in [0,\frac{2(K-kL)}{k}+1]$, present in the matrix ${\bf P}_1$, where \begin{equation}\label{submatrix AAt}
\tilde{{\bf A}} =	\left (\begin{array}{cccc} 
	{\bf A} & {\bf A}^{T}
	\end{array} \right )
	\end{equation} 
	  All the diagonal entries as well as the entries below the diagonal are stars in the matrix {\bf A}. So, for the transpose matrix ${\bf A}^{T}$, all the diagonal entries and the entries above the diagonal are stars. Hence, for any integer $s \in [0,S_1)$, suppose that the integer $s$ is present in the position corresponding to $i^{th}$ row and $j^{th}$ column of the matrix ${\bf A}$. Then, for $i,j \in [0,\frac{K-kL}{k}], j>i$,
	  \begin{equation}
	  \left( \begin{array}{cccc}
	 \tilde{a}_{i,j}&\tilde{a}_{i,i+\frac{K-kL+k}{k}}\\
	 \tilde{a}_{j,j} & \tilde{a}_{j,i+\frac{K-kL+k}{k}}
	  \end{array}\right) = \left( \begin{array}{cccc}
	  s&\star\\
	 \star & s
	  \end{array}\right) .
	  \end{equation}
	Hence for any integer $s$ present in the matrix ${\bf A}$,  the sub-matrix given by (\ref{submatrix AAt}) satisfies the condition \textit{C3} in Definition \ref{def: PDA}.
	
	For the same reason stated for the sub-matrix (\ref{submatrix AAt}), the sub-matrix $\left (\begin{array}{cccc} 
	{\bf A}^{T} \\ {\bf A}
	\end{array} \right )$ also satisfies the condition \textit{C3} in Definition \ref{def: PDA}.
	
	 Suppose we take any two non-adjacent row blocks in ${\bf P}_1$, say $i^{th}$ and $j^{th}$ row blocks, $i,j \in [0,\frac{K}{K-kL+k}),$ (the $0^{th}$ and the $(\frac{K}{K-kL+k} -1)^{th}$ row blocks are considered as adjacent row blocks). It can be observed that the sub-matrix containing only the blocks {\bf A} and ${\bf A}^{T}$ in the corresponding column blocks is of the form 
	 \begin{equation}\label{submatrix}
	 \left (\begin{array}{cccc} 
	 {\bf A} & {\bf A}^{T}&{\bf X}&{\bf X}\\
	 {\bf X}&{\bf X}&{\bf A} & {\bf A}^{T}
	 \end{array} \right ) \text{ or } 	 \left (\begin{array}{cccc} 
	 {\bf X}&{\bf X}&{\bf A} & {\bf A}^{T}\\
	  {\bf A} & {\bf A}^{T}&{\bf X}&{\bf X}
	 \end{array} \right )
	 \end{equation}

Since the matrix {\bf X} contains only stars, it is evident that the sub-matrices given by (\ref{submatrix}) satisfies the condition \textit{C3} in Definition \ref{def: PDA}.

	 Suppose we take any two adjacent row blocks in ${\bf P}_1$, say $i^{th}$ and $(i+1)^{th}$ row blocks, $i \in [0,\frac{K}{K-kL+k})$ (the $0^{th}$ and the $(\frac{K}{K-kL+k} -1)^{th}$ rows are considered as adjacent row blocks and if $i=\frac{K}{K-kL+k} -1$, then $i+1$ is considered as $0$). It can be observed that the sub-matrix containing only the blocks {\bf A} and ${\bf A}^{T}$ in the corresponding columns blocks is of the form 
	 \begin{equation}\label{submatrix adjacent}
	{\bf A}'=(a'_{i,j})= \left (\begin{array}{cccc} 
	 {\bf A} & {\bf A}^{T}&{\bf X}\\
	 {\bf X}&{\bf A} & {\bf A}^{T}
	 \end{array} \right )
	 \end{equation}
	 For any integer $s \in [0,S_1)$, suppose that the integer $s$ is present in the position corresponding to the $i^{th}$ row and the $j^{th}$ column of the matrix ${\bf A},i,j \in [0,\frac{K-kL}{k}], j>i$, then 
	 \begin{align*}
 & \left( \begin{array}{cccc}
a'_{i,j}&a'_{i,i+l} &a'_{i,j+l} &a'_{i,i+2l}\\
a'_{j,j} & a'_{j,i+l}&a'_{j,j+l}&a'_{j,i+2l}\\
a'_{i+l,j}&a'_{i+l,i+l}&a'_{i+l,j+l}&a'_{i+l,i+2l} \\
a'_{j+l,j}&a'_{j+l,i+l}&a'_{j+l,j+l} & a'_{j+l,i+2l}
\end{array}\right)\\
&=\left( \begin{array}{cccc}
s&\star &\star&\star\\
\star &s&\star&\star\\
\star&\star&s&\star\\
\star&\star&\star &s
\end{array}\right)
	 \end{align*}
	  where $l=\frac{K-kL+k}{k}$.
	 	Hence, for any integer $s$ present in the matrix ${\bf A}$,  the sub-matrix given by (\ref{submatrix adjacent}) satisfies the condition \textit{C3} in Definition \ref{def: PDA}.
	 	
	 	In short, if we take any integer $s$ in the matrix ${\bf P}_1$, the matrix ${\bf P}_1$ satisfies the condition \textit{C3} in Definition \ref{def: PDA}. All the three conditions $\textit{C1},\textit{C2}'$ and $\textit{C3}$ are satisfied by the matrix ${\bf P}_1$. Hence it is a $\frac{2K}{K-kL+k}$-regular $(\frac{K}{k},\frac{K}{k},L,\frac{(K-kL)(K-kL+k)}{2k^2} )$ PDA.
	 	
	 	Now, we need to prove that the matrix ${\bf P}_1$ is $k$-cyclic.
	 	Consider the $0^{th}$ column block in the block matrix ${\bf P}_1$:
	 	\begin{equation}
	 		 	{\bf C}_0=\left ( \begin{array}{ccccccc}
	 	{\bf A}\\ \hline{\bf X}\\\vdots\\ {\bf X}\\ \hline{\bf A}^{T}
	 	\end{array}\right ).
	 	\end{equation}
There are $\frac{K-kL+k}{k}$ columns and $\frac{K}{k}$ rows  in the matrix ${\bf C}_0$, indexed by $[0,\frac{K-kL}{k}]$ and $[0,\frac{K}{k})$ respectively. Recall that all the diagonal entries and the entries below the diagonal are stars in the matrix {\bf A}. For the transpose matrix ${\bf A}^{T}$, all the diagonal entries and the entries above the diagonal are stars. Hence the structure of stars in the matrix ${\bf C}_0$ is as follows:
	\begin{equation} \label{star structure}
\left ( \begin{array}{ccccccc}
\star& & & &\\
\star&\star & & &\\
\star& \star&\ddots& &\\
\star& \star&\ldots& \star&\\ \hline
\star& \star&\ldots& \star&\\
\star& \star&\ldots& \star&\\
\vdots& \vdots&\ldots& \vdots&\\
\star& \star&\ldots& \star&\\ \hline
\star& \star&\ldots& \star&\\
& \ddots&\star& \star&\\
&&\star& \star&\\
&&& \star&
\end{array}\right ).
\end{equation}
	 	We know that the number of the stars in each column in the matrix ${\bf P}_1$ is $L$. From (\ref{star structure}), it is clear that all the $L$ stars in each column occurs in consecutive rows. The $L$ stars in the $0^{th}$ column are present in the positions corresponding to the first $L$ rows. The structure of stars in rest of the  columns is obtained by cyclically shifting the stars in the previous column towards down by $1$ unit.
	 	We also know that in the block matrix ${\bf P}_1$, each column block is obtained by cyclically shifting the blocks in the previous column block towards down. Hence the matrix ${\bf P}_1$ represents a $1$-cyclic $\frac{2K}{K-kL+k}$-regular $(\frac{K}{k},\frac{K}{k},L,\frac{(K-kL)(K-kL+k)}{2k^2} )$ PDA.
 	\end{proof}
\begin{lem}
	The $K \times \frac{K}{k}$ matrix $\tilde{{\bf P}}_1$ obtained using \textbf{procedure 3} of Algorithm \ref{algo2} is a $k$-cyclic $\frac{2K}{K-kL+k}$-regular $(\frac{K}{k},K,kL,\frac{(K-kL)(K-kL+k)}{2k} )$ PDA.
\end{lem}

	\begin{proof}
The number of integers present in the matrix $\tilde{{\bf P}}_1$ is $\tilde{S}_1 =kS_1$ and the number of stars present in each column of the matrix $\tilde{{\bf P}}_1$, is $\tilde{Z}_1=kZ_1=kL$. 	This is because, for each $i \in [0,\frac{K}{k})$, the $(ik)^{th}$ row of the matrix $\tilde{{\bf P}}_1$ is same as the $i^{th}$ row of the matrix ${\bf P}_1$ (as in step $21$) and for the $(ik+j)^{th}$ row, $j \in [2,k)$ of the matrix $\tilde{{\bf P}}_1$, each entry is obtained by adding $S_1$ to the corresponding elements in the previous row of the matrix $\tilde{{\bf P}}_1$, where $\star+S_1 =\star$.  This is to make sure that the integers used for filling the $(ik+j)^{th}$ row is different from the integers used for filling the $(ik)^{th}$ row to $(ik+j-1)^{th}$ row. Hence, the condition $\textit{C1}$ in the Definition \ref{def: PDA} is satisfied by the matrix $\tilde{\bf P}_1$.
	
	Consider $S_1$ integers in the set $[0,S_1)$. The sub-matrix of the matrix $\tilde{{\bf P}}_1$, obtained by taking only the rows in which the integers in $[0,S_1)$ are present, is ${\bf P}_1$. Now, consider $S_1$ integers in the set $[S_1,2S_1)$. The sub-matrix of the matrix $\tilde{{\bf P}}_1$, obtained by taking only the rows in which the integers in $[S_1,2S_1)$ are present, is ${\bf P}_1 +S_1$, where  ${\bf P}_1 +S_1 =(p_{i,j} +S_1)$, with $\star + S_1=\star$. The positions of stars in ${\bf P}_1 +S_1$ is same as that of ${\bf P}_1$ and each integer entries in ${\bf P}_1 +S_1$ is obtained by adding $S_1$ to the corresponding entry in ${\bf P}_1$. Hence, each integer in $[S_1,2S_1)$ occurs $\frac{2K}{K-kL+k}$ times in the sub-matrix  ${\bf P}_1 +S_1$ and the condition $\textit{C3}$ in Definition \ref{def: PDA} is satisfied by the matrix ${\bf P}_1 +S_1$.
	Similarly for any $i \in [1,k-1]$ if we consider $S_1$ integers in the set $[iS_1,(i+1)S_1)$, The sub-matrix of the matrix $\tilde{{\bf P}}_1$, obtained by taking only the rows in which the integers in $[iS_1,(i+1)S_1)$ are present, is ${\bf P}_1 +iS_1$, where  ${\bf P}_1 +iS_1 =(p_{i,j} +iS_1)$, with $\star + iS_1=\star$. The position of stars in ${\bf P}_1 +iS_1$ is same as that of ${\bf P}_1$ and each integer entries in ${\bf P}_1 +iS_1$ entries are obtained by adding $iS_1$ to the corresponding entry in ${\bf P}_1$. Hence, if we take any $s \in [iS_1,(i+1)S_1)$, the condition $\textit{C3}$ in Definition \ref{def: PDA} is satisfied by the matrix ${\bf P}_1 +iS_1$. Also, each integer in $[iS_1,(i+1)S_1)$ occurs $\frac{2K}{K-kL+k}$ times in the sub-matrix  ${\bf P}_1 +iS_1$ and hence the condition $\textit{C2}'$ in Definition \ref{def regular PDA} is satisfied by the matrix ${\bf P}_1 +iS_1$.
	
	 Summing up all the observations, for any $s \in [0,kS_1)$, the condition $\textit{C3}$ in Definition \ref{def: PDA} is satisfied by the matrix $\tilde{{\bf P}}_1$. Also,  the condition $\textit{C2}'$ in Definition \ref{def regular PDA} is satisfied by the matrix $\tilde{{\bf P}}_1$ with $g=\frac{2K}{K-kL+k}.$ 
	 
	 We know that there are $L$ stars in each column in the matrix ${\bf P}_1$ which occurs in consecutive rows and position of stars  in each column is obtained by the shifting the position of stars in the previous column downwards by $1$ unit. We also know that for each $i \in [0,\frac{K}{k})$, the $(ik)^{th}$ row of the matrix $\tilde{{\bf P}}_1$ is same as the $i^{th}$ row of the matrix ${\bf P}_1$ and the position of stars in the  $(ik+1)^{th}$ row to  $((i+1)k-1)^{th}$ row in the matrix $\tilde{{\bf P}}_1$, is same as that of the $(ik)^{th}$ row of the matrix $\tilde{{\bf P}}_1$. Therefore, in the matrix $\tilde{{\bf P}}_1$, there are $kL$ stars in each column which occur in consecutive row and the position of stars  in each column is obtained by the shifting the position of stars in the previous column downwards by $k$ units. Hence the matrix $\tilde{{\bf P}}_1$ represents a $k$-cyclic $\frac{2K}{K-kL+k}$-regular $(\frac{K}{k},K,kL,\frac{(K-kL)(K-kL+k)}{2k} )$ PDA.

\end{proof}

{\textit{Proof of Theorem \ref{thm valid pda}}}:
Each matrix $\tilde{{\bf P}}_t, t \in [2,k]$ represents a $k$-cyclic $\frac{2K}{K-kL+k}$-regular  $(\frac{K}{k},K,kL,\frac{(K-kL)(K-kL+k)}{2k} )$ PDA with the $\frac{(K-kL)(K-kL+k)}{2k}$ integers present in the PDA $\tilde{{\bf P}}_t$ being $[(t-1)\tilde{S}_1,t\tilde{S}_1)$. 	This is because, each entry in $ \tilde{{\bf P}}_t$ is obtained by adding $(t-1)\tilde{S}_1$ to the corresponding entry in $\tilde{{\bf P}}_1$, where $\star +(t-1)\tilde{S}_1 =\star.$
	
	The number of stars present in each column of the matrix $\tilde{{\bf P}}_t$ is same as that in  $\tilde{{\bf P}}_1$. Hence, the matrix ${\bf P}$ obtained by concatenating all the matrices in $\{\tilde{{\bf P}}_t: t \in [1,k]\}$ obeys condition $\textit{C1}$ in Definition \ref{def: PDA} with $Z=\tilde{Z}_1=kL$.
	The matrix ${{\bf P}}$ is obtained by concatenating $k$ number of $k$-cyclic $\frac{2K}{K-kL+k}$-regular  $(\frac{K}{k},K,kL,\frac{(K-kL)(K-kL+k)}{2k} )$ PDAs with the integers present in each of the PDAs  $\tilde{{\bf P}}_t, t\in [1,k],$ are different from one another.
 Hence, the total number of integers present in the matrix ${\bf P}$ obtained by concatenating all the matrices in $\{\tilde{{\bf P}}_t: t \in [1,k]\}$ is $S=k\tilde{S}_1=\frac{(K-kL)(K-kL+k)}{2}$ and the matrix ${\bf P}$ obeys condition $\textit{C2}'$ in Definition \ref{def regular PDA} with $g=\frac{2K}{K-kL+k}$.  The matrix ${\bf P}$ also satisfies condition $\textit{C3}$ in Definition \ref{def: PDA}. Hence the matrix ${\bf P}$ represents a  $\frac{2K}{K-kL+k}$-regular $(K,K,kL,\frac{(K-kL)(K-kL+k)}{2} )$ PDA.

		Now, we need to prove that the PDA ${\bf P}$ is $k$-cyclic. Since, each matrix $\tilde{{\bf P}}_t, t \in[1,k]$ represents a $k$-cyclic $\frac{2K}{K-kL+k}$-regular  $(\frac{K}{k},K,kL,\frac{(K-kL)(K-kL+k)}{2k} )$ PDA, the $kL$ stars in each column in ${\bf P}$ occur in consecutive rows. The matrix  ${\bf P}$ is $k$-cyclic if the position of stars in each column is obtained by shifting the stars in the previous column towards down by $k$ units.		
		The first $\frac{K}{k}$ columns in ${\bf P}$  is $k$-cyclic since the matrix  $\tilde{{\bf P}}_1$ is $k$-cyclic. We have assumed that $k|K$, hence shifting the position of stars in the $(\frac{K}{k}-1)^{th}$ column of $\tilde{{\bf P}}_1$ towards down by $k$ units will result in the pattern of stars in the $0^{th} $ column of $\tilde{{\bf P}}_1$. The pattern of stars in the $0^{th} $ column of $\tilde{{\bf P}}_1$ and $\tilde{{\bf P}}_2$ is the same. Hence shifting the position of stars in the $(\frac{K}{k}-1)^{th}$ column of ${\bf P}$ towards down by $k$ units will result in the pattern of stars in the $(\frac{K}{k})^{th} $ column of ${\bf P}$. Similarly after every $\frac{K}{k}$ columns the pattern repeats. So, the PDA ${\bf P}$ is $k$-cyclic. This completes the proof.

	\section{Discussion}
In this work, we have constructed a new class of PDAs which we call as \textit{$t$-cyclic $g$-regular PDA}. This class of PDA is used for providing the delivery scheme for multi-access coded caching problems when $\gamma \in \left \{\frac{k}{K}: k | K, (K-kL+k)|K, k\in \left [1,K \right ]\right \}$. The sub-packetization level required for our scheme is the least compared to the state-of-the-art schemes. For certain ranges of values of $L$, the transmission rate is also less compared to some of the existing schemes. We have obtained a scheme based on PDA only for certain ranges of values of $\gamma$. Obtaining a scheme based on PDA with linear sub-packetization for the entire range of $\gamma$ is an interesting problem to work on.

\section*{Acknowledgement}
This work was supported partly by the Science and Engineering Research Board (SERB) of Department of Science and Technology (DST), Government of India, through J. C. Bose National Fellowship to B. Sundar Rajan.
 

\begin{thebibliography}{1}
	\bibitem{maddahali2015flcc}
	M. Maddah-Ali and U. Niesen, ``Fundamental limits of caching,'' \emph{IEEE Transactions on Information Theory}, vol. 60, no. 5, pp. 2856–2867, 2014.
	
	\bibitem{tuninetti2016uncoded}
	K. Wan, D. Tuninetti and P. Piantanida, ``On the optimality of uncoded cache placement,'' in \emph{Information Theory Workshop (ITW), 2016 IEEE}, pp. 161–165.
	
	
	\bibitem{avestimar2017tradeoff} 
	Q. Yu, M. A. Maddah-Ali and A. S. Avestimehr, ``The Exact Rate-Memory Tradeoff for Caching With Uncoded Prefetching,'' in \emph{IEEE Transactions on Information Theory}, vol. 64, no. 2, pp. 1281-1296, Feb. 2018.
	
	\bibitem{jin2018placeOpti} 
	S. Jin, Y. Cui, H. Liu and G. Caire, ``Uncoded placement optimization for coded delivery,'' in \emph{IEEE $16^{th}$ International Symposium on Modeling and Optimization in Mobile, Ad Hoc, and Wireless Networks (WiOpt)}, 2018, pp. 1–8.
	
	\bibitem{vilardebo2018coded} 
	J. Gómez-Vilardebó, ``Fundamental Limits of Caching: Improved Rate-Memory Tradeoff With Coded Prefetching," in \emph{IEEE Transactions on Communications}, vol. 66, no. 10, pp. 4488-4497, Oct. 2018
	
	\bibitem{nieson2017nonuniform} 
	U. Niesen and M. A. Maddah-Ali, ``Coded Caching With Nonuniform Demands," in \emph{IEEE Transactions on Information Theory}, vol. 63, no. 2, pp. 1146-1158, Feb. 2017
	
	\bibitem{gunduz2018nonuniform} 
	E. Ozfatura and D. Guenduez, ``Uncoded Caching and Cross-Level Coded Delivery for Non-Uniform File Popularity,'' 2018 \emph{IEEE International Conference on Communications (ICC)}, Kansas City, MO, 2018, pp. 1-6
	
	\bibitem{nieson2015decent} 
	M. A. Maddah-Ali and U. Niesen, ``Decentralized Coded Caching Attains Order-Optimal Memory-Rate Tradeoff," in \emph{IEEE/ACM Transactions on Networking}, vol. 23, no. 4, pp. 1029-1040, Aug. 2015
	
	\bibitem{caire2016d2d} 
	M. Ji, G. Caire and A. F. Molisch, ``Fundamental Limits of Caching in Wireless D2D Networks," in \emph{IEEE Transactions on Information Theory}, vol. 62, no. 2, pp. 849-869, Feb. 2016.
	
	
	\bibitem{Yan2017PDA} 
	Q. Yan, M. Cheng, X. Tang and Q. Chen, ``On the Placement Delivery Array Design for Centralized Coded Caching Scheme,'' in \emph{IEEE Transactions on Information Theory}, vol. 63, no. 9, pp. 5821-5833, Sept. 2017.
	
		\bibitem{cheng2020PDA} 
	M. Cheng, J. Jiang, X. Tang and Q. Yan, ``Some Variant of Known Coded Caching Schemes With Good Performance," in \emph{IEEE Transactions on Communications}, vol. 68, no. 3, pp. 1370-1377, March 2020.
	
	\bibitem{cheng2019groupingscheme} 	
	M. Cheng, J. Jiang, Q. Wang and Y. Yao, ``A Generalized Grouping Scheme in Coded Caching," in \emph{IEEE Transactions on Communications}, vol. 67, no. 5, pp. 3422-3430, May 2019.
	
	\bibitem{cheng2019flexiblememory} 		
M. Cheng, J. Jiang, Q. Yan and X. Tang, ``Constructions of Coded Caching Schemes With Flexible Memory Size," in \emph{IEEE Transactions on Communications}, vol. 67, no. 6, pp. 4166-4176, June 2019.

	\bibitem{Michel2020EdgeColoring} 		
	J. Michel and Q. Wang, ``Placement Delivery Arrays From Combinations of Strong Edge Colorings," in \emph{IEEE Transactions on Communications}, vol. 68, no. 10, pp. 5953-5964, Oct. 2020.
	
	\bibitem{SZG2018hypergraph} 
	C. Shangguan, Y. Zhang and G. Ge, ``Centralized Coded Caching Schemes: A Hypergraph Theoretical Approach," in \emph{IEEE Transactions on Information Theory}, vol. 64, no. 8, pp. 5755-5766, Aug. 2018.
	
	\bibitem{Yan2018bipartitegraph} 
	Q. Yan, X. Tang, Q. Chen and M. Cheng, ``Placement Delivery Array Design Through Strong Edge Coloring of Bipartite Graphs," in \emph{IEEE Communications Letters}, vol. 22, no. 2, pp. 236-239, Feb. 2018.
	
		\bibitem{zhong2020concatenating} 
X. Zhong, M. Cheng and J. Jiang, ``Placement Delivery Array Based on Concatenating Construction," in \emph{IEEE Communications Letters}, vol. 24, no. 6, pp. 1216-1220, June 2020
	
	\bibitem{nihkil2017multiaccess}
	J. Hachem, N. Karamchandani and S. N. Diggavi, ``Coded caching for multi-level popularity and access,'' \emph{IEEE Transactions on Information Theory,} vol. 63, no. 5, pp. 3108–3141, 2017.
	
	\bibitem{nikhil2020ratememory}
	K. S. Reddy and N. Karamchandani, ``Rate-memory trade- off for multi-access coded caching with uncoded placement,'' \emph{IEEE Transactions on Communications}, Vol. 68, No. 6, pp. 3261-3274, 2020.
	
	\bibitem{parinello2019multiacessgains}
	B. Serbetci, E. Parrinello and P. Elia, ``Multi-access coded caching: gains beyond cache-redundancy,'' \emph{2019 IEEE Information Theory Workshop (ITW),} Visby, Sweden, 2019, pp. 1-5.	
	
	\bibitem{SR2020improvedrate}
	S. Sasi, B. S. Rajan, ``An Improved Multi-access Coded Caching with Uncoded Placement,'' arXiv:2009.05377v2 [cs.IT], Oct 2020.
	
	\bibitem{Caire2020noveltransformation}
	M. Cheng, D. Liang, K. Wan, M. Zhang, and G. Caire, ``A Novel Transformation Approach of Shared-link Coded Caching Schemes for Multiaccess Networks,'' arXiv:2012.04483 [cs.IT], Dec 2020.
	
		\bibitem{nikhil2020structured}
	K. S. Reddy, N. Karamchandani, ``Structured Index Coding Problem and Multi-access Coded Caching,'' arXiv:2012.04705 [cs.IT], Dec 2020.
	
		\bibitem{MR2020linearsubpacketization}
	A. A. Mahesh, B. S. Rajan, ``A Coded Caching Scheme with Linear Sub-packetization and its Application to Multi-Access Coded Caching,'' arXiv:2009.10923 [cs.IT], Sep 2020.
		
	\bibitem{KMR2020CRD}
	Digvijay Katyal, Pooja Nayak M, and B. Sundar Rajan, ``Multi-access Coded Caching Schemes From Cross Resolvable Designs,''  \emph{IEEE Transactions on Communications}, (Available as early access article in IEEE Xplore).
	
		

\end{thebibliography}
\end{document}